\documentclass[11pt]{article}

\usepackage[utf8]{inputenc}
\usepackage[T1]{fontenc}
\usepackage{lmodern}
\usepackage[DIV=11]{typearea} % large value = less margin (roughly speaking), recommended values: 10pt: 8, 11pt: 10, 12pt: 12
\usepackage{tikz}

\usepackage{microtype}

\usepackage{amssymb}
\usepackage{amsmath}
\usepackage{amsthm}
\usepackage{thmtools}
\usepackage{thm-restate}

\usepackage{xcolor}
\usepackage{xspace}
\usepackage{xfrac}

\usepackage{subcaption}

\usepackage{comment}
\usepackage[normalem]{ulem}

\usepackage[backend=biber, style=alphabetic, backref=true, doi=false, url=false, maxcitenames=3, mincitenames=3, maxbibnames=10, minbibnames=10, sortlocale=en_US]{biblatex}  

\usepackage[ocgcolorlinks]{hyperref} % Option ocgcolorlinks makes links only colored on screen and not on printouts
\usepackage{cleveref}
\usepackage[bottom]{footmisc}

\usepackage[procnumbered,ruled,vlined,linesnumbered]{algorithm2e}

\DeclareUnicodeCharacter{221A}{$\sqrt{}$}
\DeclareUnicodeCharacter{3F5}{$\eps$}

\colorlet{DarkRed}{red!50!black}
\colorlet{DarkGreen}{green!50!black}
\colorlet{DarkBlue}{blue!50!black}

\hypersetup{
	linkcolor = DarkRed,
	citecolor = DarkGreen,
	urlcolor = DarkBlue,
	bookmarksnumbered = true,
	linktocpage = true
}

\graphicspath{ {./input/images/} }

\DontPrintSemicolon
\SetKw{KwAnd}{and}
\SetProcNameSty{textsc}
\SetFuncSty{textsc}

\newtheorem{Lem}{Lemma}
\newtheorem{theorem}{Theorem}
\newtheorem{Cor}{Corollary}

\newtheorem{Inv}{Invariant}

\usepackage{todonotes}

\newcommand{\eps}{\epsilon}

\title{Decremental APSP in Directed Graphs Versus an Adaptive Adversary}
\author{Jacob Evald\thanks{\texttt{jeav@di.ku.dk}. The author is supported by the Starting Grant 7027-00050B from the Independent Research Fund Denmark under the Sapere Aude research career programme.}
\and Viktor Fredslund-Hansen\thanks{\texttt{viha@di.ku.dk}. The author is supported by the Starting Grant 7027-00050B from the Independent Research Fund Denmark under the Sapere Aude research career programme.}
\and Maximilian Probst Gutenberg\thanks{\texttt{probst@di.ku.dk}. The author is supported by Basic Algorithms Research Copenhagen (BARC), supported by Thorup's Investigator Grant from the Villum Foundation under Grant No. 16582.}
\and Christian Wulff-Nilsen\thanks{\texttt{koolooz@di.ku.dk}, \texttt{http://www.diku.dk/$_{\widetilde{~}}$koolooz/}. The author is supported by the Starting Grant 7027-00050B from the Independent Research Fund Denmark under the Sapere Aude research career programme.}}

\date{}
\begin{document}

\maketitle

\begin{abstract}
Given a directed graph $G = (V,E)$, undergoing an online sequence of edge deletions with $m$ edges in the initial version of $G$ and $n = |V|$, we consider the problem of maintaining all-pairs shortest paths (APSP) in $G$. 

Whilst this problem has been studied in a long line of research [ACM'81, FOCS'99, FOCS'01, STOC'02, STOC'03, SWAT'04, STOC'13] and the problem of $(1+\eps)$-approximate, weighted APSP was solved to near-optimal update time $\tilde{O}(mn)$ by Bernstein [STOC'13], the problem has mainly been studied in the context of \emph{oblivious} adversaries, which assumes that the adversary fixes the update sequence before the algorithm is started. 

In this paper, we make significant progress on the problem in the setting were the adversary is adaptive, i.e.  can base the update sequence on the output of the data structure queries. We present three new data structures that fit different settings:
\begin{itemize}
    \item We first present a \emph{deterministic} data structure that maintains the \emph{exact} distances with total update time $\tilde{O}(n^3)$\footnote{We use $\tilde O$-notation to hide logarithmic factors.}. 
    \item We also present a \emph{deterministic} data structure that maintains $(1+\eps)$-approximate distance estimates with total update time $\tilde O(\sqrt{m} n^2/\eps)$ which for sparse graphs is $\tilde O(n^{2+1/2}/\eps)$.
    \item Finally, we present a randomized $(1+\eps)$-approximate data structure which works against an adaptive adversary; its total update time is $\tilde O(m^{2/3}n^{5/3} + n^{8/3}/(m^{1/3}\epsilon^2))$ which for sparse graphs is $\tilde O(n^{2+1/3})$.
\end{itemize}
Our exact data structure matches the total update time of the best \emph{randomized} data structure by Baswana et al. [STOC'02] and maintains the distance matrix in near-optimal time. Our approximate data structures improve upon the best data structures against an adaptive adversary which have $\tilde{O}(mn^2)$ total update time [JACM'81, STOC'03].
\end{abstract}

\pagebreak

\section{Introduction}\label{sec:Intro}

Shortest paths is a classical algorithmic problem dating back to the $1950$s. The two main variants are the all-pairs shortest paths (APSP) problem and the single-source shortest paths (SSSP) problem, both of which have been extensively studied in various models, including the partially and fully dynamic setting. 

A dynamic graph algorithm is an algorithm that maintains information about a graph that is subject to updates such as insertions and deletions of edges or vertices. Such a graph can model real-world networks that change over time, such as road networks where traffic changes and roads are blocked from time to time.
We say that a dynamic graph problem is \textit{decremental} if it only allows deletions,  \textit{incremental} if it only allows insertions and \textit{fully-dynamic} if it allows both. Incremental and decremental graphs are referred to as being \textit{partially-dynamic}. A dynamic graph algorithm aims to efficiently process a sequence of online updates interspersed with queries about some property of the underlying dynamic graph.

\subsection{Problem Definition}
In this paper, we consider the \emph{decremental all-pairs shortest-paths} problem where the goal is to efficiently maintain shortest path distances between all pairs of vertices in a decremental directed graph $G=(V, E)$. We shall restrict our attention to the case where $G$ is unweighted. Letting $m$ denote the initial number of edges and $n = |V|$, we want a data-structure which for any $u,v \in V$ supports the following operations:
\begin{itemize}
    \item $\textsc{Dist}(u,v)$: reports the shortest path distance $d_G(u,v)$ from $u$ to $v$ in the current version of $G$,
    \item $\textsc{Delete}(u,v)$: deletes an edge $(u,v)$ from $E$.
\end{itemize}
We furthermore consider the problem also in its relaxed version where we only aim to maintain \emph{approximate} distance estimates which can then be queried. We denote by $\tilde{d}_G(u,v)$ a distance estimate for the distance from $u$ to $v$ and we say that an APSP algorithm has an \textit{approximation ratio} (or \textit{stretch}) of $t > 1$ if for any $u,v \in V$, we have that $d_G(u,v) \leq \tilde{d}_G(u,v) \leq t \cdot d_G(u,v)$. This paper will be concerned with both the exact and the $(1+\epsilon)$-approximate version of the problem. 

Another focus of this article is the \textit{adversarial model}; the adversarial model defines the model under which the sequence of updates and queries are assumed to be made by an \textit{adversary}. We say that a performance guarantee of an algorithm works against an \emph{oblivious} adversary if the adversary must define the sequence of updates before the algorithm starts for the guarantee to hold. Thus the sequence of updates is independent of any random bits used by the algorithm. This is opposed to algorithms that work against an \emph{adaptive} adversary, where the adversary is allowed to create the update sequence 
``on the go'', e.g. based on answers to previous queries made to the data structure. Depending on the data structure, these choices may not be independent on the random choices made, which may result in the data structure performing poorly. One key advantage of a data structure that works against an adaptive adversary is that it can be used inside an algorithm as a black box, regardless of whether that algorithm adapts its updates to answers to queries. We point out that \emph{deterministic} data structures always work against an adaptive adversary.

The performance of a \emph{partially-dynamic} algorithm is usually measured in terms of the \textit{total update time}. That is, the accumulated time it takes to process all updates (edge deletions). The \textit{query time}, on the other hand, is the time to answer a single distance query. A natural goal is to minimize the total update time while keeping the stretch and query time small. Since all the structures presented in this paper explicitly maintain a distance matrix, the query time is constant.

% \subsection{Motivation}
% The problem of maintaining all-pairs shortest paths in decremental graph is well-motivated by its application to fully-dynamic network modeling where the decremental data structure can be used as an internal data structure using a standard reduction from fully-dynamic to the decremental setting (see for example \cite{henzinger2016dynamic}). 

% Moreover, shortest path problems are known to appear as subroutines in various static algorithms  \cite{madry2010faster, alstrup2017constructing, bernstein2018online, ancona2018algorithms, DBLP:journals/corr/abs-1812-01602, Chuzhoy:2019:NAD:3313276.3316320}. For example, in \cite{alstrup2017constructing}, the authors design an algorithm to construct light spanners in near-linear time using a partially-dynamic data structure that maintains approximate all-pairs shortest path distances. Moreover, Mądry showed in \cite{madry2010faster} that the problem of computing multi-commodity flows can be reduced to solving decremental {weighted} approximate all-pairs shortest paths. We point out that both algorithms require the data structures to work against an adaptive adversary, thus motivating our central objective of obtaining efficient data structures for this setting.

\subsection{Prior Work}
The naive approach to dynamic APSP is to recompute the shortest path distances after each update using the best static algorithm. The query time is then constant and the time for a single update is $\tilde{O}(mn)$ for APSP and $\tilde{O}(m)$ for SSSP. At the other end of the spectrum one could achieve optimal update time by simply updating the input graph and only running an SSSP algorithm whenever a query is processed. Running a static algorithm each time, however, fails to reuse any information between updates whatsoever and gives a high query time, motivating more efficient dynamic approaches that do this.

In 1981, Even and Shiloach \cite{shiloach1981line} gave a deterministic data-structure for maintaining a shortest path tree to given depth $d$ in an undirected, unweighted decremental graph in total time $O(md)$. Henzinger and King \cite{henzinger2005bicon} and King \cite{king1999fully} later adapted this to directed graphs with integer weights. Running their structure for each vertex solves the decremental all-pairs shortest paths problem in $O(mn^2W)$ time, where edge weights are integers in $[1,W]$. 

Henzinger and King were the first to improve upon this bound, giving an algorithm with total update time $\tilde{O}(m n^{2.5} \sqrt{W})$ \cite{king1999fully} which is an improvement for $W = \omega(n)$. Demetrescu and Italiano \cite{demetrescu2006fully} improved this data structure slightly and showed that the restriction to integral edge weights can be removed. Finally, the same authors \cite{demetrescu2004new} presented a data structure with total update time $\tilde{O}(mn^2)$ which is the state of the art for any data structure against an adaptive adversary up to today. In fact, their algorithm can be extended to a fully-dynamic algorithm with $\tilde{O}(n^2)$ amortized update time and which can handle vertex updates\footnote{In this case, vertex updates refers to insertions or deletions of vertices with up to $n-1$ incident edges.}. We also point out that this data structure was later simplified and generalized by Thorup \cite{thorup2004fully}.

Around the same time Baswana, Hariharan, and Sen \cite{baswana2002improved} gave an \emph{oblivious} Monte-Carlo construction with total update time $\tilde{O}(n^3)$ for \emph{unweighted} graphs. Further, they showed that their data structure could be adapted to give an $(1+\eps)$-approximate APSP algorithm for \emph{weighted} graphs with total update time of $\tilde{O}(\sqrt{m}n^2/\varepsilon)$. Finally, Bernstein presented a $(1+\eps)$-approximate algorithm with total running time $\tilde{O}(mn \log W / \eps)$ by using a clever approach of shortcutting paths \cite{bernstein2016maintaining}. Whilst his algorithm achieves near-optimal running time, again, the algorithm has to assume an oblivious adversary.

More recently, Karczmarz and Łącki \cite{karczmarz2020simple} gave a deterministic $(1+\epsilon)$-approximate APSP algorithm for decremental graphs that runs in total time $\tilde O(n^3 \log W/\epsilon)$. They also presented the first non-trivial algorithm for incremental graphs \cite{karczmarz2019reliable} achieving total update time $\tilde O(mn^{4/3}\log W/\epsilon)$.

We refer the reader to Appendix \ref{sec:relatedWork} for a more comprehensive treatment of related work which also includes algorithms for \emph{undirected} graphs and algorithms with larger stretch.

\subsection{Our Contributions}
In this paper, we present three new data structures for the all-pairs shortest paths problem. Our first theorem gives a \emph{deterministic} data structure for the exact variant of the problem with near-optimal $\tilde O(n^3)$ total update time. It also matches the best \emph{randomized} algorithm by Baswana et al. \cite{baswana2002improved} and constitutes a significant improvement over the previous best bound of $\tilde O(mn^2)$ which is obtained by running an ES-tree \cite{shiloach1981line} from every source or by the data structure Italiano et al. \cite{demetrescu2004new} and improves over all but the sparsest graph densities. Our data structure is near-optimal as we will show an $\Omega(n^3)$ lower bound on the total update time of any decremental data structure that explicitly maintains the distance matrix.

\begin{theorem}\label{Thm:DetExact}
Let $G$ be an unweighted directed graph with $n$ vertices and initially $m$ edges. Then there exists a deterministic data structure which maintains all-pairs shortest path distances in $G$ undergoing an online sequence of edge deletions using a total time of $O(n^3\log^3n)$. The $n\times n$ distance matrix is explicitly maintained so that at any point, a shortest path distance query can be answered in constant time. The data structure can report a shortest path between any query pair in time proportional to the length of the path. 
\end{theorem}

Our second result is concerned with maintaining $(1+\epsilon)$-approximate all-pairs shortest path distances. This constitutes the first deterministic data structure that solves the problem in subcubic time with small approximation error (except for graphs that are not extremely dense). In fact, for very sparse graphs with $m = \tilde{O}(n)$, our update time even matches the near-optimal result by Bernstein \cite{bernstein2016maintaining} with total update time $\tilde{O}(mn)$.

\begin{theorem}\label{Thm:DetApprox}
Let $G$ be an unweighted directed graph with $n$ vertices and initially $m$ edges. Then given $\epsilon > 0$, there exists a deterministic data structure that maintains all-pairs $(1+\epsilon)$-approximate shortest path distances in $G$ undergoing an online sequence of edge deletions using a total time of $O(\sqrt m n^2\log^2n/\epsilon)$. At any point, a $(1+\epsilon)$-approximate shortest path distance query can be answered in constant time and a $(1+\epsilon)$-approximate shortest path between the query pair can be reported in time proportional to the length of the path.
\end{theorem}

Our third result gives a data structure achieving a better time bound. While we use randomization to achieve the improved time bound, our algorithm again works against an adaptive adversary.

\begin{theorem}\label{Thm:Rand}
Let $G$ be an unweighted directed graph with $n$ vertices and initially $m$ edges. Then given any $\epsilon > 0$, there exists a Las Vegas data structure that maintains all-pairs $(1+\epsilon)$-approximate shortest path distances in $G$ under an online sequence of edge deletions using a total expected time of $\tilde O(m^{2/3}n^{5/3}/\epsilon + n^{8/3}/(m^{1/3}\epsilon^2))$. This bound holds w.h.p.~and the data structure works against an adaptive adversary. At any point, a $(1+\epsilon)$-approximate shortest path distance query can be answered in constant time.
\end{theorem}

We summarize our results as well as previous state-of-the-art results in Table \ref{tab:results}.

\begin{table}[ht]
\renewcommand{\arraystretch}{1.2}
\centering
\begin{tabular}{|p{3cm}|p{3cm}|p{4.5cm}|p{3cm}|}
\hline
Time & Approximation & Adversary/ Deterministic & Reference \\ \hline
$O(mn^2)$ &  exact & deterministic  & \cite{shiloach1981line, demetrescu2004new} \\ \hline
$\tilde{O}(n^3)$ & exact & deterministic  & \textbf{New Result} \\ \hline
$\tilde{O}(n^3)$ & exact & adaptive  &
\cite{baswana2002improved} \\ \hline
\end{tabular}

\caption{\label{tab:results} Our results and previous state-of-the-art results for decremental APSP in the exact setting.}
\end{table}

\begin{table}[ht]
\renewcommand{\arraystretch}{1.2}
\centering
\begin{tabular}{|p{3cm}|p{3cm}|p{4.5cm}|p{3cm}|}
\hline
Time & Approximation & Adversary/ Deterministic & Reference \\ \hline
$\tilde{O}(\sqrt{m} n^2 / \eps)$ & $(1+\eps)$ & deterministic & \textbf{New Result} \\ \hline
$\tilde{O}(m^{2/3}n^{5/3}/\epsilon + n^{8/3}/(m^{1/3}\epsilon^2))$ & $(1+\eps)$ & adaptive & \textbf{New Result} \\ \hline
$\tilde{O}(\sqrt{m} n^2 / \eps)$ & $(1+\eps)$ & oblivious & \cite{baswana2002improved} \\ \hline
$\tilde{O}(nm)$ & $(1+\eps)$ & oblivious  & \cite{bernstein2016maintaining} \\ \hline
\end{tabular}

\caption{ Our results and previous state-of-the-art results for decremental APSP in the approximate setting.}
\end{table}

\subsection{Overview}\label{subsec:techOverview}
%\todo[inline]{CW: I think this technical overview goes into too much detail with heavy notation, precise inequalities, and so on. I think it should be somewhat shorter (maybe one page or so) and more high-level. Also, I added a high-level description in the beginning of Section 5; that structure is by far the most technical and we should probably only give a small part of the intuition here in the intro - maybe focus only on the idea that separators are always kept small by using random sampling and when there are no more sampled vertices to use, the remaining set of candidate separator vertices is small in expectation (see high-level description in Sec. 5) and then make it clear to the reader that more intuition will be given in Section 5. The figure is good and fits nicely here in the intro although we might want to avoid the math in the right-hand side of the figure.}
Our overall approach for the deterministic data structures is similar to that of Baswana et al.~\cite{baswana2002improved} but with a key difference that allows us to avoid using a randomized hitting set and instead rely on deterministic separators. The idea of the construction by Baswana et al.~relies on a well-known result which says that if we sample a subset $H^{\rho}_i$ of the vertices of size $\tilde O(n/\rho^i)$ (where $\rho$ is some constant strictly larger than $1$), each with uniform probability, then, w.h.p. we "hit" each shortest-path of length $[\rho^i, \rho^{i+1})$ between any pair of vertices in any version of the graph $G$.

Phrased differently, given vertices $u,v \in V$, we have that if the the shortest path from $u$ to $v$ is of length $\ell \in [\rho^i, \rho^{i+1})$, then there is some vertex $w \in H^{\rho}_i$, such that the concatenation of the shortest path from $u$ to $w$ and the shortest path from $w$ to $v$ is of length $\ell$. For each such $w$, we say $w$ is a \emph{witness} for the tuple $(u,v)$ for distance $\ell$. 

Now for each $u,v \in V$, if the initial distance from $u,v$ was $\ell \in [\rho^i, \rho^{i+1})$, we can check $H^{\rho}_i$ to find a witness $w$. If the length of the path from $u$ to $w$ to $v$ is increased, we can continue our scanning of $H^{\rho}_i$ to see whether another witness exists. If there is no witness $w \in H^{\rho}_i$ left at some stage, we know that there is no path of length $\ell$ left in $G$ w.h.p. and increase our guess by setting $\ell \mapsto \ell + 1$. 

Sampling initially a hitting set $H^{\rho}_i$ for every $i \in [0, \log_{\rho} n]$, we can find the "right" hitting set for each distance $\ell$. Observe now that for each tuple $(u,v) \in V^2$, we have to scan a hitting set of size $\tilde O(n/\rho^i)$ for $\rho^{i+1} - \rho^i \sim \rho^{i+1}$ levels before the hitting set index $i$ is increased which only occurs $O(\log n)$ times, thus we only spend time $\tilde O(n)$ for each vertex tuple $(u,v)$. Thus, the total running time of the searches for witnesses can be bound by $\tilde O(n^3)$. 

\paragraph{The Deterministic Exact Data Structure}
Our construction is similar in the sense that we maintain witnesses for each distance scale $[\rho^i, \rho^{i+1})$ for every $i \in [0,\log_{\rho} n]$ such that each distance $\ell$ is in one such distance scale. The key difference is that instead of using a randomized \emph{global} hitting set $H^{\rho}_i$ for a distance scale $[\rho^i, \rho^{i+1})$, our construction relies on deterministically maintaining a small local vertex separator $S_i(u)$ for every vertex $u \in V$ of size $\tilde{O}(n/\rho^i)$ separating all shortest paths starting in $u$ with a distance in $[\rho^i, \rho^{i+1})$.

More precisely, for each distance scale $[\rho^i, \rho^{i+1})$ and vertex $u \in V$, we maintain a separator $S_i(u)$ that satisfies the invariant that every shortest path from $u$ to a vertex $v$ at distance at least $\rho^i$ is intersected by a vertex in $S_i(u)$. If this invariant is violated after an adversarial update, then we find such a vertex $v$ and need to add additional vertices to $S_i(u)$ during the time step. The challenge is to take these additional separator vertices such that the total size of $S_i(u)$ is not increased beyond $\tilde{O}(n/\rho^i)$. We defer the details of the separator procedure to a later section and continue our discussion of the APSP data structure.

Since we need to detect whether vertices have distance less than $\rho^i$ from $u$ or not in $G$, we further have to use a bottom-up approach to compute distances, i.e. we start with the smallest possible distance range and find all small distances and then find larger distances using the information already computed. This issue did not arise in  Baswana et al.~\cite{baswana2002improved} but can be handled by a careful approach.

It is now easy to see that the scanning for witnesses can be implemented in the same time as in the analysis sketched above by scanning the list of local separator vertices which serve as witnesses instead of the hitting set. Further, we can maintain local vertex separators using careful arguments in total time $\tilde O(mn)$ giving our result in Theorem~\ref{Thm:DetExact}.

\begin{figure}[!hbt]
\centering
\includegraphics[width=0.80\textwidth]{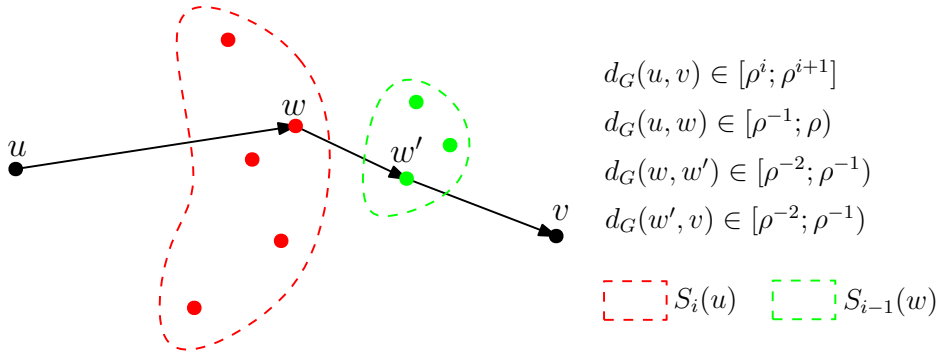}
\caption{Illustration of separators and path ``hierarchy''. Here $u \leadsto v$ goes through a witness $w$, and $w \leadsto v$ goes through $w'$. If the length of the path $w' \leadsto v$ is increased by $\Delta$, the distance estimates of all 2-hop-paths that use $w' \leadsto v$ as a sub-path are increased by that amount. In this case, the estimate for $w \leadsto w' \leadsto v$ is increased and is propagated to the next level where subsequently the estimate for $s \leadsto w \leadsto v$ is increased.}
\end{figure}

\paragraph{The Deterministic Approximate Data Structure} In order to improve the running time for sparse graphs, we can further focus on only considering distances that are roughly at a $(1+\epsilon)$-multiplicative factor from each other. More concretly, instead of increasing the expected distance from $\ell$ to $\ell + 1$ when we cannot find a witness for some path from $u$ to $v$ for distance $\ell$, we can increase the next expected distance level $\ell'$ to $\sim (1+\epsilon)\ell$ and consider every vertex $w$ a witness if there is a path $u \leadsto w \leadsto v$ of length at most $\ell'$. Thus, we search less distances and can thereby increase the time to maintain distances that are at least $d$ in total time $\tilde{O}(n^3/d + mn)$. Again, a careful approach is necessary to ensure that approximations do not add up over distance scales. 

This is no faster than the data structure for exact distances when $d$ is small so in order to get Theorem~\ref{Thm:DetApprox}, we use the $O(mnd)$ data structure of Even and Shiloach \cite{shiloach1981line} to maintain distances up to $d$. Picking $d$ such that $mnd = n^3/d$ gives the result of Theorem~\ref{Thm:DetApprox} (the term $\tilde{O}(mn)$ vanishes since it is subsumed by the two other terms, also we assumed $\epsilon > 0$ to be a constant to simplify the presentation). 

\paragraph{Maintaining Separators}
We now describe how to deterministically maintain the ``small'' local separator for a vertex $s \in V$ with some useful invariants. 

Let $S$ be the local separator for $s$. The first invariant that will be useful is that any vertex $t \in V$ that is reachable from $s$ in $G \setminus S$, is ``close'' to $s$ or roughly within distance $d$. As edges are deleted from $G$, the distances from $s$ to vertices in $S$ may increase, meaning that separator vertices may move away from $s$ as edges are deleted. When some separator vertex $w \in S$ moves too far away, the invariant is re-established by growing BFS trees in parallel, one layer at a time, from $s$ in $G \setminus S$ and from $w$ in $G \setminus S$ with the edges reversed. The search halts when a layer (corresponding to the leaves of the BFS tree at the current iteration) that is ``thin'' is found, and its vertices are added to $S$; vertices that are on the opposite side of the separator than $s$ are cut off as they must all be too far away from $s$. Here, "thin" refers to a BFS layer such that the number of vertices added to the separator is only a factor $\tilde{O}(1/d)$ times the number of vertices cut off. It is well known that such a layer exists (cfr. Lemma $\ref{Lem:ThinLayer}$ for the details). Summing up, it follows that $|S| = \tilde{O}(n/d)$ at all times. By marking vertices as they are searched (according to the side of the BFS layer on which they are found), the vertices that are ``cut off'' from $s$ by the augmented separator will never be searched again, and the cost of searching the edges of either side of the search can be charged to sum of the degree of these vertices, for a total update time of $O(m)$.

For our randomized data structure, we need an additional property that essentially allows us to take a snapshot of the current separator and use it in later updates rather than having to repeatedly update the separator. This will be key to getting an improved randomized time bound. Details can be found in Lemma~\ref{Lem:DynSep} which states our separator result.
%Consider $t \in V$ for which $s$ and $t$ are roughly at distance $d$ at some time step $t_0$. Then in any subsequent time step $t_1$ for which $s$ is still roughly at distance $d$, any shortest $s$-to-$t$ path comprises subpaths $s \leadsto w$ and $w \leadsto t$ for which $d_G(s,w) \leq d$ and $d_G(w,t) \leq d$, where $w$ is the first vertex on the path belonging to the separator $S$ at timestep $t_0$. The invariant implies that as long as $s$ and $t$ remain roughly at distance $d$, the shortest $s$-to-$t$ path distance is given by $d_G(s,w) + d_G(w,t)$ for $w \in S'$ minimizing this quantity where $S'$ is the separator at time step $t_0$. This essentially allows us to take a "snapshot" of the separator and use that as the separator for shortest paths from $s$ to $t$ as long as they remain close enough.

\paragraph{The Randomized Approximate Data Structure} 
The randomized approximate data structure of Theorem~\ref{Thm:Rand} follows the same overall approach but is technically more involved. Instead of keeping track of all $2$-hop paths $u\leadsto s\leadsto v$ for every $s\in S_i(u)$, the randomized data structure samples a subset of these by picking each vertex of $S_i(u)$ independently with some probability $p$. It only keeps track of approximate shortest path distances through this subset rather than the full set $S_i(u)$. This will speed up the above since priority queue sizes are reduced in expectation by a factor $p$. However, this approach fails once no short $2$-hop path intersects the sampled subset. At this point, w.h.p.~there should only be short $2$-hop paths through $O(\log n/p)$ vertices of $S_i(u)$ so also in this case, the priority queue sizes can be kept small. However, scanning linearly through $S_i(u)$ to find this small subset will take $\tilde O(n/d)$ time and over all pairs $(u,v)$.

Our solution is roughly the following. Suppose no sampled vertices certify an approximate short path from $u$ to $v$. Then $v$ scans linearly through $S_i(u)$ to find the small $O(\log n/p)$ size subset $S_i'(u)$. Consider the set $W$ of vertices $w$ such that $d_G(w,v)$ is small compared to $d$, i.e., $d_G(w,v)\leq \epsilon d$ for some small constant $\epsilon > 0$. Then we show that the small subset $S_i'(u)$ found for $v$ can also be used for each vertex $w\in W$. The intuition is that for any vertex $s\in S_i(u)\setminus S_i'(u)$, the approximate shortest path distance from $u$ to $w$ through $s$ must be large since otherwise we get a short path $u\leadsto s\leadsto w\leadsto v$ from $u$ to $v$ through $s$, contradicting that $s\notin S_i'(u)$.

It follows that if $|W|$ is large, the $\tilde O(n/d)$ cost of scanning $S_i(u)$ can be distributed among a large number of vertices of $W$. Dealing with the case where $|W|$ is small is more technical so we omit it here.

The way we deal with an adaptive adversary is roughly as follows. Consider a deterministic data structure that behaves like the randomized data structure above, except that it maintains $2$-hop paths $u\leadsto s\leadsto v$ for all $S_i(u)$ rather than only through a sampled subset. The slack from the approximation allows us to round up all ``short'' approximate distances to the same value. Hence, as long as the randomized data structure has short $2$-hop paths, it maintains exactly the same approximate distances as the deterministic structure and hence the approximate distances output to the adversary is independent of the random bits used.

%\paragraph{Extending to Edge Weights.} Finally, we point out that we believe that our approximate data structures can be extended to deal with edge weights. Therefore, we \emph{decouple} the distances from the \emph{hop}. Let us therefore review the notion of \emph{hops} and define an $h$-hop-restricted shortest path from $s$ to $t$ in $G$ to be the path from $s$ to $t$ consisting of at most $h$ edges with minimal weight. 

%Let us now reconsider our deterministic $(1+\eps)$-approximate data structure. We first change the Even-Shiloach data structures to work with weighted graphs using a simple edge rounding scheme. As shown in \cite{bernstein2016maintaining}, we can maintain a $(1+\eps)$-approximate estimate of the $d$-hop-restricted shortest path from each source $s$ in time $\tilde{O}(md \log W)$. Using the underlying approximate shortest-path trees, we can now find separators using the layers defined by the number of edges used such that either a vertex $v$ has a $(1+\eps)$-approximate shortest path consisting of at most $d$ edges, or the separator of size $\tilde{O}(n/d)$ contains at least one node on the shortest path from $s$ to $v$. Analogously, we can use this argument in the hierarchy to find the hop $(1+\delta)d$-hop-restricted shortest paths using the $d$-hop-restricted shortest paths for some $\delta > 0$. A similar argument can be used to extend our adaptive data structure.

\section{Definitions and Notation}\label{sec:DefsNotation}
In the following, let $G = (V,E)$ be a directed unweighted graph. The graph $G_{\mathbf{rev}}$ is obtained from $G$ by reversing the orientation of each edge. For any two vertices $u,v\in V$, we denote by $u \leadsto v$ a shortest path from $u$ to $v$ in $G$ and let $d_G(u,v)$ denote the distance of such a path. We extend this notation to sets so that, e.g., $d_G(u,V') = \min\{d_G(u,v)\vert v\in V'\}$ for $V'\subseteq V$. 

We define a BFS-layer to mean the set of nodes at some fixed distance from some $v$ in $G$. An \emph{in-tree} in $G$ is a BFS tree in $G_{\mathbf{rev}}$.

We will need notation to refer to dynamically changing data at specific points in time. Consider a sequence of updates to some object $X$ where each update takes place at a time step $t \in \mathbb N$. We denote by $X^{(t)}$ the object just after update $t$. Here, $X$ could be a graph, a shortest path distance, etc.

For handling small distances, we rely on the data-structure of Even and Shiloach \cite{shiloach1981line}, the properties of which we will state in the following lemma:

\begin{Lem}[\cite{shiloach1981line}]\label{lem:estree}
Given a directed unweighted graph $G$ undergoing a sequence of edge deletions, a source vertex $s \in V$, and $d>0$, a shortest path tree in $G$ rooted at $s$ can be maintained up to distance $d$ in total time $O(md)$. The structure requires $O(m)$ space and can be constructed in time $O(m + n)$.
\end{Lem}

\section{Maintaining Separators}\label{sec:Sep}
Lemma~\ref{Lem:DynSep} below provides a key tool used in all of our data structures. It gives an efficient data structure that maintains a growing separator set $S$ of small size in a decremental graph $G$ with the following guarantees. Let $s$ be a fixed vertex and let $d$ be some given threshold distance. Then at every time step, vertices reachable from $s$ in $G\setminus S$ are of distance slightly less than $d$ from $s$ in $G$. Conversely, for vertices $v$ not reachable from $s$ in $G\setminus S$, we have $d_G(s,v) = \Omega(d)$; furthermre, if $d_G(s,v)$ is larger than $d$ by some small constant factor then any shortest path $s\leadsto v$ in $G$ can be decomposed into $s\leadsto w\leadsto v$ such that $s\in S$, $d_G(s,w)\leq d$, and $d_G(w,v)\leq d$. In fact, the lemma states that $w$ can be chosen in $S^{t_0}$ where $t_0$ is the first time step in which $d_G(s,v)$ became (slightly) larger than $d$; note that this is a stronger statement since $S$ is growing over time. Before proving Lemma~\ref{Lem:DynSep}, we need the following well-known result.
\begin{Lem}\label{Lem:ThinLayer}
Given a directed unweighted $n$-vertex graph $G = (V,E)$, given $d_1,d_2\in\mathbb N_0$ with $d_2 - d_1 + 1\geq \lg n$ , and given vertices $u,v\in V$ with $d_G(u,v)\geq d_2$, a BFS tree in $G$ with root $u$ contains a layer $L\subseteq V$ with $d_1\leq d_G(u,L)\leq d_2$ and $|L|\leq |L_-|\lg n/(d_2-d_1+1)$ where $L_- = \{w\in V\vert d_G(u,w) < d_G(u,L)\}$ is the union of layers closer to $u$ than $L$.
\end{Lem}
\begin{proof}
Denote by $L_i$ the $i$th layer of the BFS tree from $u$. For each $i$, let $L_{{}<i} = \cup_{j < i}L_j$. Let $q = (d_2-d_1 + 1)/\lg n$. Assume for contradiction that $L$ does not exist. Then for $i = d_1,\ldots,d_2$, $|L_i| > |L_{{}<i}|/q$ so $|L_{{}<i+1}| = |L_i| + |L_{{}<i}| > (1+1/q)|L_{{}<i}|$. Since $q\geq 1$, we have $(1+1/q)^q\geq 2$ so
\[
  |L_{{}<d_2+1}| > (1+1/q)^{d_2 - d_1 + 1}|L_{{}<d_1}|\geq 2^{(d_2-d_1 + 1)/q} = n,
\]
contradicting that there are only $n$ vertices in $G$.
\end{proof}

\begin{Lem}\label{Lem:DynSep}
Given a directed unweighted $n$-vertex graph $G = (V,E)$ undergoing a sequence of edge deletions, a source $s\in V$, and a value $d\in\mathbb N$ with $d > 33\lg n$. Let $\mathcal O$ be a data structure that maintains for each $v\in V$ a distance estimate $\tilde d(s,v)\geq d_G(s,v)$ such that if $d_G(s,v)\leq d$ then $\tilde d(s,v)\leq \frac 4 3 d_G(s,v)$. Whenever an estimate $\tilde d(s,v)$ grows to a value of at least $\frac{32}{33}d$, $\mathcal O$ outputs $v$. Then there is a data structure $\mathcal D$ with access to $\mathcal O$ which maintains a growing set $S\subseteq V$ such that for each $v\in V$,
\begin{enumerate}
\item at the end of each update, if $v$ is reachable from $s$ in $G\setminus S$ then $d_G(s,v) < \frac{32}{33}d$ and otherwise $d_G(s,v) > \frac 2 3 d$,
\item if $t_0$ is a time step in which $d < d_G^{(t_0)}(s,v) \leq \frac{34}{33}d$ then for every time step $t_1\geq t_0$ in which $d_G^{(t_1)}(s,v) \leq \frac{34}{33}d$, any shortest $s$-to-$v$ path $P$ in $G^{(t_1)}$ intersects $S^{(t_0)}$ and for the first such intersection vertex $w$ along $P$, $d_G^{(t_1)}(s,w)\leq d$, and $d_G^{(t_1)}(w,v)\leq d$.
\end{enumerate}
At any time, $|S| = O(n\log n/d)$ and $\mathcal D$ has total update time $O(m)$, excluding the time spent by $\mathcal O$.

%$(1+\epsilon)$-approximate distances from $s$ in $G$ up to distance $d_2$ which reports Then there is a deterministic data structure which maintains a growing set $S\subseteq V$ in total time $O(m)$. At any time, every shortest path $P = s\leadsto t$ from $s$ in $G$ of length at least $d_1$ and length at most $d_2$ intersects $S$ and the suffix of $P$ of length $d_1$ does not intersect $S$.
\end{Lem}
\begin{proof}
Let $\epsilon = \frac 1{33}$. For each $v\in V$, let $\hat d(v)$ be obtained from the degree of $v$ in the initial graph $G$ by rounding up to the nearest multiple of $\Delta = \lceil m/n\rceil$. In the description of $\mathcal D$ below, processing one edge takes at most one unit of time.

Data structure $\mathcal D$ initializes $S = \emptyset$ and unmarks all vertices of $V$. Whenever $\mathcal O$ outputs an unmarked vertex $v$, $\mathcal D$ runs a modified BFS from $s$ in $G_S = G\setminus S$ which for each vertex $w$ spends $\hat d(w)$ time to process its outgoing edges; this can always be achieved by busy-waiting at $w$ if needed. In parallel, $\mathcal D$ runs a similar modified BFS from $v$ in $G_S' = (G\setminus S)_{\mathbf{rev}}$. The search from $s$ halts if a layer $L_s$ is found such that $\frac 2 3 d < d_{G_S}(s,L_s)\leq (\frac 2 3 + \epsilon)d$ and $|L_s| = O((x\log n)/d)$ where $x$ is the number of vertices visited by the search, excluding $L_s$. Similarly, the search from $v$ halts if a layer $L_v$ is found such that $d_{G_S'}(v,L_v) < \epsilon d$ and $|L_v| = O((y\log n)/d)$ where $y$ is the number of vertices visited by the search excluding $L_v$. Let $L$ be the first of the two layers found. $\mathcal D$ halts both searches when $L$ is found. Then $L$ is added to $S$. The existence of $L$ follows from Lemma~\ref{Lem:ThinLayer} which applies since by assumption, $\epsilon d > \lg n$.

Observe that when $\mathcal O$ outputs $v$, we have $d_G(s,v)\geq (1-\epsilon)d/(4/3) = (\frac 2 3 + 2\epsilon)d$ as otherwise, $\tilde d(s,v) < 1 - \epsilon = \frac{32}{33}d$. This shows the existence of $L_s$ and $L_v$ and that no edge is visited by both searches. We have $d_{G_S}(s,v)\geq d_G(s,v)\geq (\frac 2 3 + 2\epsilon)d$ and $d_{G_S}(s,L_s) > \frac 2 3 d$ and for every $w\in L_v$,
\begin{align*}
  d_{G_S}(s,w) &\geq d_{G_S}(s,v) - d_{G_S}(w,v) \\
  &\geq \left(\frac 2 3 + 2\epsilon \right)d - d_{G_S'}(v,w) \\
  &= \left(\frac 2 3 + 2\epsilon\right)d - d_{G_S'}(v,L_v) \\
  &> \left(\frac 2 3 + \epsilon \right)d,
\end{align*}
implying that $d_{G_S}(s,L_v) > (\frac 2 3 + \epsilon)d$. It follows that $d_{G_S}(s,L) = \min\{d_{G_S}(s,L_s),d_{G_S}(s,L_v)\} > \frac 2 3 d$.

\paragraph{Showing part $1$:}
Let $v\in V$ and consider any point during the sequence of updates. Assume first that $v$ is reachable from $s$ in $G_S$. Then $\mathcal O$ has not yet output $v$ (otherwise, the above procedure separates $v$ from $s$ with $S$) so $d_G(s,v)\leq\tilde d(s,v) < \frac{32}{33}d$, as desired.

Now, assume that $v$ is not reachable from $s$ in $G_S$. We may assume that there is a shortest path $P$ from $s$ to $v$ in $G$ since otherwise $d_G(s,v) = \infty > \frac 2 3 d$. Let $w$ be the first vertex of $S$ along $P$. It suffices to show that $|P| > \frac 2 3 d$. At some earlier point in time, the procedure added $w$ to $S$; just prior to this, $P$ was contained in $G_S$ so from the above $|P| > \frac 2 3 d$, as desired.

\paragraph{Showing part $2$:}
Let $t_0\leq t_1$ satisfy the second part of the lemma. Since $d_G^{(t_0)}(s,v) > d$ by assumption, the first part of the lemma implies that $v$ is not reachable from $s$ in $G_S^{(t_0)}$ and hence $v$ is also not reachable from $s$ in $G_S^{(t_1)}$.

Let $P$ be a shortest path from $s$ to $v$ in $G^{(t_1)}$. From what we have just shown, $P$ must intersect $S^{(t_0)}$. Let $w$ be the first vertex of $S^{(t_0)}$ along $P$. Then clearly, $d_G^{(t_1)}(s,v) = d_G^{(t_1)}(s,w) + d_G^{(t_1)}(w,v)$. Since the vertex $w'$ preceding $w$ on $P$ is reachable from $s$ in $G_S^{(t_0)}$, the first part of the lemma implies that $d_G^{(t_0)}(s,w)\leq d_G^{(t_0)}(s,w') + 1 < \frac{32}{33}d + 1$ and $d_G^{(t_0)}(s,w) > \frac 2 3 d$. The latter implies that $d_G^{(t_1)}(w,v) = d_G^{(t_1)}(s,v) - d_G^{(t_1)}(s,w)\leq \frac{34}{33}d - d_G^{(t_0)}(s,w) < \frac{34}{33}d - \frac 2 3 d < d$, showing one of the two inequalities in the second part of the lemma.

%$d_G^{(t_0)}(w,v) \leq d_G^{(t_1)}(w,v) = d_G^{(t_0)}(s,v) - d_G^{(t_0)}(s,w) < \frac{34}{33}d - \frac 2 3 d < \frac{32}{33}d$.

We show the other inequality by contradiction so assume that $d_G^{(t_1)}(s,w) > d$. Then $d_G^{(t_1)}(s,w)\geq d+1$ so by the above $d_G(s,w)$ would have increased by more than $d+1 - (\frac{32}{33}d + 1) = \frac 1{33}d$ from time step $t_0$ to $t_1$. Combining this with $d_G^{(t_1)}(s,v) = d_G^{(t_1)}(s,w) + d_G^{(t_1)}(w,v)$, $d_G^{(t_0)}(w,v)\leq d_G^{(t_1)}(w,v)$, and the triangle inequality, we get
\[
d_G^{(t_1)}(s,v) - d_G^{(t_0)}(s,v)\geq d_G^{(t_1)}(s,w) + d_G^{(t_1)}(w,v) - (d_G^{(t_0)}(s,w) + d_G^{(t_0)}(w,v)) > \frac 1{33}d
\]
This contradicts the assumption $d < d_G^{(t_0)}(s,v)\leq d_G^{(t_1)}(s,v) \leq \frac{34}{33}d$. We conclude that $d_G^{(t_1)}(s,w)\leq d$ and $d_G^{(t_1)}(w,v)\leq d$ which shows the second part of the lemma.

\paragraph{Bounding $|S|$ and running time:}
To bound, $|S|$, consider the two parallel searches from $s$ and from $v$, respectively, in some update. As argued earlier, there cannot be an edge visited by both searches. Let $X$ resp.~$Y$ be the set of vertices visited by the BFS from $s$ resp.~$v$, excluding $L_s$ resp.~$L_v$ and let $x = |X|$ and $y = |Y$.

Assume first that $L = L_s$. Then all vertices in $Y\cup L_v$ become unreachable in $G_S$ once $L$ has been added to $S$. Since $\hat d(w)/\Delta\geq 1$ for each $w\in V$, since each BFS spends $\hat d(w)$ time to process edges incident to each vertex $w$, and since the two searches run in parallel, we have
\[
  |L| = O((x\log n)/d) = O\left(\frac{\log n}d\sum_{w\in X}\frac{\hat d(w)}{\Delta}\right) = O\left(\frac{\log n}d\sum_{w\in Y\cup L_v}\frac{\hat d(w)}{\Delta} \right)
\]

Now, assume that $L = L_v$. Then all vertices of $Y\cup L_v$ become unreachable in $G_S$ once $L$ has been added to $S$ so again,
\[
  |L| = O((y\log n/d) = O\left(\frac{\log n}d\sum_{w\in Y\cup L_v}\frac{\hat d(w)}{\Delta}\right)
\]

In both cases, $|L|$ can be paid for by charging each vertex $w$ no longer reachable from $s$ in $G_S$ a cost of $O(\frac{\log n}d\hat d(w)/\Delta)$. Since a vertex is only charged once during the course of the algorithm, we get that for the final separator $S$ (and hence for each intermediate separator)
\begin{align*}
  |S| &= O\left(\frac{\log n}d\sum_{w\in V}\frac{\hat d(w)}{\Delta}\right) \\
  &= O\left(\frac{\log n}d\sum_{w\in V}\frac{d(w) + \Delta}{\Delta}\right) \\
  &= O\left(\frac{\log n(m + n\lceil m/n\rceil)}{d\lceil m/n\rceil}\right) \\
  &= O\left(\frac{n\log n}d\right)
\end{align*}
where the last bound follows since we may assume that all vertices are initially reachable from $s$ in $G$, implying $m\geq n - 1$ and hence $\lceil m/n\rceil = \Theta(m/n)$. This shows the desired bound on $|S|$.

%Since the two searches occur in parallel, the number of such vertices is at least the number $x$ of vertices visited by the search from $s$. We charge the $|L_s| = O((x\log n)/d)$ cost of adding $L_s$ to $S$ evenly among these $x$ vertices. Conversely, if $L = L_v$ then we charge the $|L_v| = O((y\log n)/d)$ cost of adding $L_v$ to $S$ evenly among the $y$ vertices visited by the search from $v$.

%Observe that each vertex of $V$ is charged at most once since when a vertex is charged, it becomes unreachable from $s$ in $G_S$. Since each vertex is charged a cost of $O(\log n/d)$, it follows that $|S| = O(n\log n/d)$, as desired.

The running time cost of any two parallel searches can be charged to the total degree of the vertices that become unreachable from $s$ in $G_S$ after extending $S$ with $L$. This shows that the total running time of parallel searches over all updates is $O(m)$, as desired.
\end{proof}

\section{Deterministic Decremental APSP}\label{sec:Det}
In this section, we present our deterministic data structures for the exact resp.~$(1+\epsilon)$-approximate decremental APSP problem and show Theorems~\ref{Thm:DetExact} and~\ref{Thm:DetApprox}. In the following, let $G = (V,E)$ denote the decremental graph.

\subsection{Exact distances}\label{subsec:DetExact}
Let $\rho = \frac{34}{33}$ and $D_i = \rho^i$ for $i = 0,\ldots,\lfloor\log_{\rho}n\rfloor$. For each $i$ and each $u\in V$, we give a data structure $\mathcal D_i(u)$ which for any query vertex $v$ maintains a value $\tilde d_i(u,v)\geq d_G(u,v)$ with equality if $d_G(u,v)\in(D_i,D_{i+1}]$. In each update, these data structures will be updated in order of increasing $i$.

Handling all-pairs shortest path distances up to at most $33\lg n$ can be done in $O(mn\log n)$ using the data structure of Even and Shiloach so we only consider $i$ such that $D_i\geq 33\lg n$. This allows us to apply Lemma~\ref{Lem:DynSep}. Consider such an $i$ and assume that we already have data structures for all values smaller than $i$.

Data structure $\mathcal D_i(u)$ maintains a separator set $S_i(u)$ using an instance $\mathcal S_i(u)$ of the data structure of Lemma~\ref{Lem:DynSep} with $s = u$, $d = D_i$, and with $\mathcal D_{i-1}(u)$ playing the role of $\mathcal O$. At the beginning of each update, $\mathcal S_i(u)$ updates $S_i(u)$. Then for each $v$, if $\mathcal O$ reports that $\tilde d(u,v)$ has increased from a value of at most $D_i$ to a value strictly greater than $D_i$, $\mathcal D_i(u)$ sets $S_i(u,v)$ equal to the current separator set $S_i(u)$; $\mathcal D_i(u)$ then sets up a priority queue $Q_i(u,v)$ where elements are all $s\in S_i(u,v)$ with corresponding keys $\tilde d_{i-1}(u,s) + \tilde d_{i-1}(s,v)$. During updates, whenever $\mathcal D_{i-1}(u)$ resp.~$\mathcal D_{i-1}(s)$ reports that $\tilde d_{i-1}(u,s)$ resp.~$\tilde d_{i-1}(s,v)$ increases, the key value of $s$ in $Q_i(u,v)$ increases by the same amount.

For each vertex $v$, $\mathcal D_i(u)$ maintains $\tilde d_i(u,v)$ as the min key value in $Q_i(u,v)$. This completes the description of each structure $\mathcal D_i(u)$.

%For $i = 0,\ldots,\lfloor\log_{\rho}n\rfloor$, we denote by $\mathcal D_i$ the data structure that, given a query pair $(u,v)$ queries $\mathcal D_i(u)$ with vertex $v$.

The overall data structure $\mathcal D$ maintains a priority queue $Q(u,v)$ for each vertex pair $(u,v)$ with an element for each $i$ with key value $\tilde d_i(u,v)$. For $i$ in increasing order, $\mathcal D$ updates $\mathcal D_i(u)$ for each $u$. Whenever a data structure $\mathcal D_i(u)$ increases a value $\tilde d_i(u,v)$, the corresponding key in $Q(u,v)$ is increased accordingly. On a query $(u,v)$, $\mathcal D$ reports the min key value in $Q(u,v)$.

\iffalse
Algorithm 1 shows pseudocode for the deletion procedure with some details abstracted away. Each data structure $\mathcal{D}_i(u)$ is assumed to be updated in the background according to how the distance estimates are updated in the algorithm, and not referred to in the pseudocode. Similarly, the local ES-trees are assumed to be updated together line 1 be deleting the edge $e$ from each local tree. In lines 2-3 each set $F_{0,u}$ is initialized by querying the ES-tree associated with $u$ for vertices whose distance label has changed due to the deletion of $e$ from $G$. The loop in lines 4-26 iterates across all levels s.t. at the beginning of each iteration, the set $F_{i-1,u}$ contains all vertices $v \in V$ whose distance estimate $\tilde d_{i-1}(u,v)$ was increased at level $i-1$. For each $v \in V$ for which $\tilde d_{i-1}(u,v)$ increases to more than $D_i$ for the first time (line 9), the procedure takes a snapshot of its separator, sets up a priority queue and adds all separator vertices to the queue with corresponding keys. The distance estimate $\tilde d(u,v)$ is updated accordingly with the value of the smallest key of the queue, and the vertex $v$ is added to $F_{i,v}$ (lines 10-19). Finally the keys of the 2-hop paths that depend on the increased distances on the next level are updated (lines 20-26). Here, $\Delta_i(u,v)$ is the amount with which $\tilde d(u,v)$ key increased on the level below.

\input{dec-apsp/input/content/exact-pseudo.tex}
\fi

\paragraph{Correctness:}
Consider a vertex pair $(u,v)$ at any time step $t_1$ in the sequence of edge deletions. If $d_G^{(t_1)}(u,v) = \infty$ then correctness is clear so assume otherwise and pick $i$ such that $d_G^{(t_1)}(u,v)\in (D_i,D_{i+1}]$ where $D_i\geq 33\lg n$. Let $t_0\leq t_1$ be the first time step such that $d_G^{(t_0)}(u,v)\in (D_i,D_{i+1}]$. Note that $S_i(u,v) = S_i(u)^{(t_0)}$. By the second part of Lemma~\ref{Lem:DynSep} combined with the observation that no key value in $Q_i(u,v)$ is below $d_G(u,v)$, it follows that the min key value in $Q_i(u,v)$ equals $d_G^{(t_1)}(u,v)$. This shows correctness.

\paragraph{Running time:}
Consider an $i\in\{0,\ldots,\lfloor\log_{\rho}n\rfloor\}$ with $D_i\geq 33\lg n$ and a vertex $u\in V$. We will show that maintaining $\mathcal D_i(u)$ takes $O(n^2\log^2n)$ time using a standard binary heap. Total time over all $i$ and $u$ will thus be $O(n^3\log^3n)$. This dominates the $O(n^3\log^2n)$ time to maintain priority queues $Q(u,v)$ and the $O(mn\log n)$ time for the data structure of Even and Shiloach for small values of $i$.

Maintaining $S_i(u)$ takes a total of $O(m)$ time by Lemma~\ref{Lem:DynSep}. The total number of elements in priority queues $Q_i(u,v)$ over all $v\in S_i(u,v)$ is $O(n^2\log n/D_i)$, again by Lemma~\ref{Lem:DynSep}. The number of increase-key operations for a single priority queue element $s$ of $Q_i(u,v)$ is $O(D_i)$ which takes a total of $O(D_i\log n)$ time. Over all elements of priority queues $Q_i(u,v)$, this is $O(n^2\log^2n)$.

\paragraph{Lower bound:}
We show that any data structure that explicitly maintains the distance matrix of $G$ during the sequence of deletions must use $\Omega(n^3)$ time.

Let the initial graph $G$ consist of a simple path $v_1\rightarrow v_2\rightarrow\cdots\rightarrow v_n$ augmented with edges $e_i = (v_i,v_{i+2})$ for $i = 1,3,5,\ldots,n$ (assuming $n$ is odd; otherwise, $i = 1,3,5,\ldots,n-1$). Deleting the edges not on the simple path in any order, say, by increasing index, results in $\Theta(n^2)$ vertex pairs each increasing their pairwise distance $\Theta(n)$ times. Hence, there are $\Omega(n^3)$ changes to the distance matrix, showing the lower bound.

Note that our choice of $G$ for the lower bound is sparse; it is straightforward to extend the above to any edge density: simply take the above graph and arbitrarily insert additional edges to reach the desired density. Then consider a sequence starting with the deletion of these additional edges followed by the sequence above.

\paragraph{Reporting paths:}
It is easy to extend our data structure to efficiently answer queries for shortest paths (rather than shortest path distances) between any vertex pair $(u,v)$. Associated with the min element of $Q(u,v)$ is a vertex $s$ such that for the associated index $i$, $\tilde d_i(u,v) = d_G(u,v) = d_G(u,s) + d_G(s,v)$, $\tilde d_{i-1}(u,s) = d_G(u,s)$, and $\tilde d_{i-1}(s,v) = d_G(s,v)$. Hence, by recursively querying for pairs $(u,s)$ and $(s,v)$, we get a shortest $u$-to-$v$ path in $G$ in time proportional to its length.

We have shown our first main result, Theorem~\ref{Thm:DetExact}.

\subsection{Approximate distances}\label{subsec:DetApprox}
Let $\epsilon > 0$ be given. We now present our deterministic data structure for the $(1+\epsilon)$-approximate variant of the problem.

The data structure is quite similar to the one for the exact variant so we only describe the changes needed. For $i > 0$ and $u\in V$, we describe data structure $\mathcal D_i(u)$ and assume that we have data structures for values less than $i$. As before, we only consider $i$ with $D_i\geq 33\lg n$.

Let $\epsilon' > 0$ be a value depending on $\epsilon$ such that $(1+\epsilon')^c = \rho$ for some $c\in\mathbb N$; we will specify $\epsilon'$ later. For $j = 0,\ldots,c = \log_{1+\epsilon'}\rho$, let $d_{i,j} = D_i(1+\epsilon')^j$. This partitions each interval $(D_i,D_{i+1}]$ into $c$ sub-intervals $(D_i(1+\epsilon')^j,D_i(1+\epsilon')^{j+1}]$ for $j = 0,\ldots,c-1$.

$\mathcal D_i(u)$ maintains $S_i(u)$ as in the exact version. For each $v\in V$, $\mathcal D_i(u)$ maintains an initially empty set $S_i(u,v)$. Once $\mathcal D_{i-1}(u)$ reports that $\tilde d_{i-1}(u,v)$ increased from a value of at most $D_i(1+\epsilon')^i$ to a value strictly greater than $D_i(1+\epsilon')^i$, $\mathcal D_i(u)$ sets $S_i(u,v)$ equal to the current set $S_i(u)$.

 For each $j = 0,\ldots,c-1$, a data structure $\mathcal D_{i,j}(u)$ maintains approximate distances $\tilde d_{i,j}(u,v)$ for each $v$ as follows. The following set is maintained:
 
 $$Q_{i,j}(u,v) = \left\{ s \in S_i(u, v) \mid \tilde d_{i-1}(u,s) + \tilde d_{i-1}(s,v) \leq (1+\epsilon')^id_{i,j} \right\}$$
 
 For ease of analysis, $Q_{i,j}(u,v)$ is maintained as a queue in which every $s \in Q_{i,j}(u,v)$ has key $\tilde d_{i-1}(u,s) + \tilde d_{i-1}(s,v)$ and is removed from $Q_{i,j}(u,v)$ (or increased to $\infty$) when this value exceeds  $(1+\epsilon')^id_{i,j}$.
 
 For each vertex $v$, $\tilde d_{i,j}(u,v) = (1+\epsilon')^id_{i,j}$ if $Q_{i,j}(u,v)$ contains at least one element and otherwise $\tilde d_{i,j}(u,v) = \infty$.

% For each $j = 0,\ldots,c-1$, a data structure $\mathcal D_{i,j}(u)$ maintains approximate distances $\tilde d_{i,j}(u,v)$ for each $v$ as follows. It maintains a priority queue $Q_{i,j}(u,v)$ containing all $s\in S_i(u,v)$ with corresponding key values $\tilde d_{i-1}(u,s) + \tilde d_{i-1}(s,v)$. For each vertex $v$, $\tilde d_{i,j}(u,v) = (1+\epsilon')^id_{i,j}$ if the min key value of $Q_{i,j}(u,v)$ is at most $(1+\epsilon')^id_{i,j}$ and otherwise $\tilde d_{i,j}(u,v) = \infty$.

Data structure $\mathcal D_i(u)$ maintains for each $v$ a min-priority queue $Q_i(u)$ with an element of key value $\tilde d_{i,j}(u,v)$ for each $j$. On query $v$, it outputs $\tilde d_i(u,v) = \min\{k,\tilde d_{i-1}(u,v)\}$ where $k$ is the min-key of this queue, i.e., $\tilde d_i(u,v) = \min\{\tilde d_{i-1}(u,v),\min_j \tilde d_{i,j}(u,v)\}$.

The overall data structure $\mathcal D$ works in the same manner as for the exact data structure. 

\iffalse
Pseudocode describing the algorithm is shown in Algorithm \ref{alg:approx}.\todo[inline]{Samme problem med "Algorithm ??" her.}

\input{dec-apsp/input/content/approx-pseudo.tex}
\fi

%priority queues as before but instead of maintaining value $\tilde d_i(u,v)$ as the min key value of $Q_i(u,v)$, it maintains it as that min key value rounded up to the nearest multiple of $(1+\epsilon')$ where $\epsilon'$ will be specified later. This way, the number of times that $\tilde d_i(u,v)$ increases over the course of the sequence of edge deletions is only $O(\log_{1+\epsilon'}n)$.

\paragraph{Correctness:}
Consider any point during the sequence of edge deletions. We will show that for suitable choice of $\epsilon'$, the estimate $\tilde d(u,v)$ that $\mathcal D$ outputs satisfies $d_G(u,v)\leq\tilde d(u,v)\leq (1+\epsilon)d_G(u,v)$ for every vertex pair $(u,v)$.% $Q_i(u,v)\neq\emptyset$.

We first show that $d_G(u,v)\leq\tilde d(u,v)$. It suffices to prove by induction on $i\geq 0$ that $d_G(u,v)\leq\tilde d_i(u,v)$. The proof holds for small $i$ such that $D_i < 33\lg n$ since then we use the data structure of Even and Shiloach, implying $\tilde d_i(u,v) = d_G(u,v)$. Now, consider an $i$ such that $D_i \geq 33\lg n$ and assume that the claim holds for smaller values than $i$. Since $\tilde d_i(u,v) = \min\{\tilde d_{i-1}(u,v),\min_j \tilde d_{i,j}(u,v)\}$ and $\tilde d_{i,j}(u,v)\geq (1+\epsilon')^id_{i,j}\geq\tilde d_{i-1}(u,s) + \tilde d_{i-1}(s,v)$, the induction hypothesis implies $\tilde d_i(u,v)\geq d_G(u,v)$, showing the induction step. Thus, $d_G(u,v)\leq\tilde d(u,v)$.

To show that $\tilde d(u,v)\leq (1+\epsilon)d_G(u,v)$, we prove by induction on $i\geq 0$ that during all updates and for all vertex pairs $(u,v)$, if $d_G(u,v)\in (0,D_{i+1}]$ then $\tilde d_i(u,v)\leq (1+\epsilon')^id_G(u,v)$. If we can show this then picking $\epsilon' \leq \ln(1+\epsilon)/(\lfloor\log_{\rho}n\rfloor)$ gives $\tilde d(u,v)\leq (1+\epsilon')^{\lfloor\log_\rho n\rfloor}d_G(u,v)\leq e^{\epsilon'\lfloor\log_\rho n\rfloor}d_G(u,v)\leq (1+\epsilon)d_G(u,v)$ for every vertex pair $(u,v)$.

We only need to consider $i$ with $D_i\geq 33\lg n$ since otherwise, we use the data structure of Even and Shiloach. Assume inductively that the claim holds for values less than $i$.

Let $t_1$ be the current time step and consider a vertex pair $(u,v)$ with $d_G^{(t_1)}(u,v)\in (0,D_{i+1}]$. By the induction hypothesis, we may assume that $d_G^{(t_1)}(u,v)\in (D_i,D_{i+1}]$. We may further assume that $\tilde d_{i-1}^{(t_1)}(u,v) > D_i(1+\epsilon')^i$ since otherwise,
\[
\tilde d_i^{(t_1)}(u,v)\leq \tilde d_{i-1}^{(t_1)}(u,v)\leq D_i(1+\epsilon')^i < (1+\epsilon')^id_G^{(t_1)}(u,v).
\]

Let $t_0\leq t_1$ be the first time step where $\tilde d_{i-1}^{(t_0)}(u,v) > D_i(1+\epsilon')^i$. We must have $d_G^{(t_0)}(u,v) > D_i$ since otherwise, the induction hypothesis would imply $\tilde d_{i-1}^{(t_0)}(u,v)\leq d_G^{(t_0)}(u,v)(1+\epsilon')^{i-1}\leq D_i(1+\epsilon')^{i-1}$, contradicting the choice of $t_0$. Since also $d_G^{(t_0)}(u,v)\leq d_G^{(t_1)}(u,v)\leq D_{i+1}$, Lemma~\ref{Lem:DynSep} implies that there is a vertex $s\in S_i^{(t_0)}(u) = S_i^{(t_0)}(u,v) = S_i^{(t_1)}(u,v)$ such that $d_G^{(t_1)}(u,v) = d_G^{(t_1)}(u,s) + d_G^{(t_1)}(s,v)$, $d_G^{(t_1)}(u,s)\leq D_i$, and $d_G^{(t_1)}(s,v)\leq D_i$.

Pick $j$ such that $d_G^{(t_1)}(u,v)\in (d_{i,j},d_{i,j+1}]$. By the induction hypothesis,
\[
  \tilde d_{i-1}^{(t_1)}(u,s) + \tilde d_{i-1}^{(t_1)}(s,v)\leq (1+\epsilon')^{i-1}d_G^{(t_1)}(u,v)\leq (1+\epsilon')^{i-1}d_{i,j+1} = (1+\epsilon')^id_{i,j}.
\]
Hence, $Q_{i,j}(u,v)$ is non-empty at time step $t_1$ so
\[
  \tilde d_i^{(t_1)}(u,v)\leq\tilde d_{i,j}^{(t_1)}(u,v) = (1+\epsilon')^{i}d_{i,j}\leq (1+\epsilon')^{i}d_G^{(t_1)}(u,v).
\]
This shows the induction step.

\paragraph{Running time:}
The analysis is similar to the one for exact distances. Pick an $i\in\{0,\ldots,\lfloor\log_\rho n\rfloor\}$ with $D_i\geq 33\lg n$. The total time to maintain $S_i(u)$ over all $u$ is $O(mn)$.

Observe that each approximate distance $\tilde d_{i-1}(u_1,u_2)$ is of the form $(1+\epsilon')^{i'}d_{i',j}$ for $i'\leq i-1$. Since each element $s$ in a queue $Q_{i,j}(u,v)$ has key value $\tilde d_{i-1}(u,s) + \tilde d_{i-1}(s,v)$, it follows that the number of increase-key operations applied to $s$ in $Q_{i,j}(u,v)$ is $O(\log_{1+\epsilon'}D_i) = O(\log D_i/\epsilon') = O(\log n/\epsilon')$. For our purpose, a simplified queue $Q_{i,j}(u,v)$ suffices which keeps a counter of the number of elements of key value at most $(1+\epsilon')^id_{i,j}$; this follows since the min key value is at most $(1+\epsilon')^id_{i,j}$ if and only if the counter is strictly greater than $0$. Every queue operation for $Q_{i,j}(u,v)$ can then be supported in $O(1)$ time. The number of elements in $Q_{i,j}(u,v)$ over all $u$, $v$, and $j$ is $O(cn^3\log n/D_i) = O(n^3\log n/(D_i\epsilon'))$ by Lemma~\ref{Lem:DynSep}. This gives a total time bound of $O(mn + n^3\log^2n/(D_i(\epsilon')^2))$. This dominates the time spent on maintaining priority queues $Q_i(u)$.

Recall from above that $\epsilon' \leq \ln(1+\epsilon)/(\lfloor\log_{\rho}n\rfloor)$. The only additional constraint on $\epsilon'$ is that $(1+\epsilon')^c = \rho$ for some $c\in\mathbb N$. This can be achieved with $\epsilon' = \Theta(\ln(1+\epsilon)/(\lfloor\log_{\rho}n\rfloor))$. Hence, we get a time bound of $O(mn + n^3\log^4n/(D_i\epsilon^2))$.

Note that this bound is no better than the exact data structure for small $D_i$. We thus consider a hybrid data structure that only applies our data structure when $D_i$ is above some distance threshold $d$ and otherwise applies the data structure of Even and Shiloach which takes a total of $O(mnd)$ time. Summing over all $D_i > d$ and applying a geometric sums argument, the total time for our hybrid data structure is
\[
  O(mnd + \sum_{i: D_i > d}n^3\log^4n/(D_i\epsilon^2))) = O(mnd + n^3\log^4n/(d\epsilon^2)))
\]
Setting $d = n\log^2n/(\epsilon\sqrt m)$ gives Theorem~\ref{Thm:DetApprox}. Showing the bound for reporting approximate shortest paths in the theorem is done in the same way as in Section~\ref{subsec:DetExact}.

\section{Randomized Decremental APSP}\label{sec:Rand}
In this section, we provide a randomized $(1+\epsilon)$-approximate data structure and show Theorem~\ref{Thm:Rand}. The data structure is Las Vegas and works against an adaptive adversary. In contrast, the data structures of \cite{baswana2002improved} and \cite{bernstein2016maintaining} are both Monte Carlo and can only handle an oblivious adversary.

\subsection{High-level description}
%\todo[inline]{This could use a figure to e.g. illustrate the in-trees, $L$-sets, etc.}
    We start by giving a high-level description of our data structure and sketch its analysis. We focus our attention on maintaining approximate distances close to the value $D_i$ from a single vertex $u$ and for now we assume an oblivious adversary.
    
    \paragraph{Maintaining a sampled separator subset:} Instead of maintaining each separator $S_{i,j}(u,v)$ (with associated with priority queue $Q_{i,j}(u,v)$) as the full vertex separator $S_i(u)$, we obtain a speed-up by only maintaining a sampled subset of $S_i(u)$. As long as this sampled subset certifies that there is a short two-hop path from $u$ to $v$, the data structure proceeds as in the previous section. When this is no longer the case, there might still be a short two-hop path from $u$ to $v$ through a non-sampled vertex $s$ in the full separator set $S_i(u)$. However, since there are no more sampled candidates, the expected number of vertices of $S_i(u)$ that provide a short two-hop path is small and we can update $S_{i,j}(u,v)$ to be this small subset.

It follows that $S_{i,j}(u,v)$ can be kept small at all times, which is needed to give a speed-up.

\paragraph{A speed-up using shallow in-trees:} The problem with the data structure sketched above is that the entire set $S_i(u)$ had to be scanned in order to update $S_{i,j}(u,v)$ which means that the data structure will not be faster than our deterministic structure from the previous section. To deal with this, consider the following modification. The set $S_{i,j}(u,v)$ is updated as before by scanning over the entire set $S_i(u)$. Now, an in-tree $T(v)$ is grown from $v$ of radius at most $\epsilon' D_i$. Each vertex $v'$ in $T(v)$ then inherits the set of $v$, i.e., $S_{i,j}(u,v')$ is updated to the set $S_{i,j}(u,v)$ and this update is fast since $S_{i,j}(u,v)$ is small in expectation. This works since $v$ is a proxy for $v'$ in the sense that a short two-hop path from $u$ to $v'$ is also a short two-hop path from $u$ to $v$ (as $T(v)$ is an in-tree of small radius). Now, the time spent on the single scan of $S_i(u)$ can be distributed among all vertices of $T(v)$ and the number of such vertices must be at least $\epsilon' D_i + 1$ (if not, $v$ would be within distance $\epsilon' D_i$ from $u$).

Unfortunately, the time analysis for the above procedure breaks down if the in-trees grown during the sequence of updates overlap too much. We now sketch how to deal with this. Mark vertices of each in-tree grown so far. When the BFS procedure grows a new in-tree $T(v)$, this procedure is modified by having it backtrack at previously marked vertices which thus become leaves of $T(v)$; this set of marked leaves will be referred to as $L$ in the detailed descripton below.

\paragraph{Case $1$, dealing with a large in-tree:} If the number of unmarked vertices visited in $T(v)$ is greater than $\epsilon' D_i$, the above procedure and analysis can be applied; this is referred to as Case $1$ in the detailed description below.

\paragraph{Case $2$, dealing with a small in-tree:} Otherwise, we are in Case $2$; here we recall that $T(v)$ has small radius and observe that the only way to enter $T(v)$ from $G\setminus T(v)$ is through $L$. Hence, for every vertex $s$ in the union $\cup_{v'\in L}S_{i,j}(u,v')$, there is a good two-hop path from $u$ to $v$ through $s$. But since we know that there is only a small number of such vertices left (in expectation), this union must be small. Furthermore, the union must contain a good separator for every vertex in $T(v)$ (again because $T(v)$ has small radius and because $T(v)$ must be entered through $L$) and we thus have an efficient way to update $S_{i,j}(u,w)$ for all $w\in T(v)$.

\paragraph{Handling an adaptive adversary:} Above we assumed an oblivious adversary. When the adversary is adaptive, we need to be more careful since the approximate distances reported might reveal information about which vertices have been sampled. To deal with this, we round up every two-hop distance on a given distance scale to the same upper bound value (this will only increase the weight of each two-hop path by a small factor so that the output to a query will still be $(1+\epsilon)$-approximate). Hence, the rounded up approximate weight of a two-hop path $u\leadsto s\leadsto v$ is the same for every choice of "good" separator vertex $s$ regardless of whether it was sampled or not. It follows that our randomized structure outputs the same distance estimates as a slower deterministic algorithm that maintains the full separator sets. Hence, the updates done by the adversary does depend on answers to previous approximate distance queries, as desired.

This completes the high-level description of our data structure.

%mn/p = n^3p/d <=> p = sqrt{md}/n,
%Time: sqrt m n^2/sqrt d = mnd <=> d = (n/sqrt m)^{2/3} = n^{2/3}/m^{1/3},
%Time: mnd = m^{2/3}n^{5/3}

\subsection{The data structure}
We now make the above formal. First, redefine $\rho = \frac{34 - \frac 1 2}{33} = \frac{67}{66}$ and pick $\epsilon'$ such that $(1+\epsilon')^c = \rho$ for some $c\in\mathbb N$ and such that $\rho(1+\epsilon')\leq\frac{34}{33}$. For each $u$ and $i$ such that $D_i\geq 33\lg n$, a separator $S_i(u)$ is maintained with a data structure $\mathcal S_i(u)$ as in Section~\ref{sec:Det}.% Once $\mathcal D_{i-1}(u)$ reports that $\tilde d_{i-1}(u,v) > D_i(1+\epsilon')^{i+2}$, $\mathcal D_i(u)$ sets $S_i(u,v)$ equal to the current set $S_i(u)$.

We extend the range of index $j$ by $1$ so that $j\in\{0,\ldots,c + 1\}$. Each structure $\mathcal D_{i,j}(u)$ maintains a growing set $M_{i,j}(u)$ of marked vertices; this set is initially empty. In the following, let $U_{i,j}(u) = V\setminus M_{i,j}(u)$ denote the set of unmarked vertices and let $G_{U_{i,j}(u)}$ denote the graph with vertex set $V$ and containing the edges of $G$ having at least one unmarked endpoint.

In each update, $\mathcal D_{i,j}(u)$ maintains $S_{i, j}(u, v)$ and $Q_{i,j}(u,v)$ for $v\in V$ in the following way.

For each $v \in V$ and every vertex $s$ added to $S_i(u)$ in the current update, $s$ is added to $S_{i, j}(u, v)$ with some probability $p$ to be fixed later. Note that only vertices $v$ for which $s$ is actually added to $S_{i,j}(u,v)$ need to be processed. In Appendix \ref{app:sampling}, we employ a different sampling scheme that avoids having to flip a coin for every vertex $v \in V$ in every update. 

For vertices $v$ such that $v\in M_{i,j}(u)$ or $v\in U_{i,j}(u)$ and $\tilde d_{i-1}(u,v)\leq D_i(1+\epsilon')^{2i}$, no further processing is done.

 Now, assume that $v\in U_{i,j}(u)$ and that $\tilde d_{i-1}(u,v) > D_i(1+\epsilon')^{2i}$. If this inequality did not hold in the previous update, each vertex of $S_{i, j}(u, v)$ is added to a new min-queue $Q_{i,j}(u,v)$ with key values as in the previous section. Conversely, if the inequality did hold in the previous update, each new vertex added to $S_{i, j}(u, v)$ in the current update is added to $Q_{i,j}(u,v)$.

% Cristians version commented out here:
% In each update, $\mathcal D_{i,j}(u)$ processes each $v\in V$ in some arbitrary order. We describe the processing of one such $v$.

% If $v\in M_{i,j}(u)$ or if $v\notin M_{i,j}(u)$ and $\tilde d_{i-1}(u,v)\leq D_i(1+\epsilon')^{2i}$, no processing is done for $v$.

% Now, assume that $v\notin M_{i,j}(u)$ and that $\tilde d_{i-1}(u,v) > D_i(1+\epsilon')^{2i}$. If this inequality did not hold in the previous update, each vertex of $S_i(u)$ is sampled independently with some probability $p$ to be fixed later. Min queue $Q_{i,j}(u,v)$ is initialized with these sampled vertices with key values as in the previous section. Conversely, if the inequality did hold in the previous update, each new vertex added to $S_i(u)$ in the current update is sampled independently with probability $p$ and the sampled vertices are added to $Q_{i,j}(u,v)$.

%Since $S_{i, j}(u, v)$ is no longer guaranteed to be the entire separator $S_i(u)$ at any timestep, but rather a sampled subset, it can happen that $S_{i, j}(u, v)$ and therefore $Q_{i,j}(u,v)$ does not contain any vertex $s$ that achieves the desired distance of $\tilde d_{i-1}(u,s) + \tilde d_{i-1}(s,v) \leq d_{i,j}(1+\epsilon')^{2i}$ but such a vertex might exist in $S_i(u) \setminus S_{i, j}(u, v)$.

If the min key value of $Q_{i,j}(u,v)$ is greater than $d_{i,j}(1+\epsilon')^{2i}$, $\mathcal D_{i,j}(u)$ grows an in-tree $T(v)$ from $v$ in $G_{U_{i,j}(u)}$ up to radius $\epsilon'D_i$.

There are now two cases: $|V(T(v))\setminus M_{i,j}(u)| > \epsilon'D_i$ and $|V(T(v))\setminus M_{i,j}(u)|\leq \epsilon'D_i$.

\begin{description}
\item[Case 1:] If $|V(T(v))\setminus M_{i,j}(u)| > \epsilon'D_i$ then $\mathcal D_{i,j}(u)$ scans once over $S_i(u)$ to find the subset of vertices $s\in S_i(u)$ for which $\tilde d_{i-1}(u,s) + \tilde d_{i-1}(s,v)\leq d_{i,j}(1+\epsilon')^{2i}$. For each $v'\in V(T(v))$, $Q_{i,j}(u,v')$ is set to contain exactly this subset of vertices $s$ but with key value $\tilde d_{i-1}(u,s) + \tilde d_{i-1}(s,v')$.
\item[Case 2: ] If $|V(T(v))\setminus M_{i,j}(u)| \leq \epsilon'D_i$ then let $L = V(T(v))\cap M_{i,j}(u)$ and let $Q = \cup_{v'\in L}Q_{i,j}(u,v')$. For each $v'\in V(T(v))\setminus L$, $\mathcal D_{i,j}(u,v)$ sets $Q_{i,j}(u,v')$ to contain the subsets of elements $s\in Q$ with $\tilde d_{i-1}(u,s) + \tilde d_{i-1}(s,v)\leq d_{i,j}(1+\epsilon')^{2i}$; their key values are $\tilde d_{i-1}(u,s) + \tilde d_{i-1}(s,v')$.
\end{description}

In both cases, $\mathcal D_{i,j}(u)$ then marks all vertices of $T(v)$, i.e., $M_{i,j}(u) \leftarrow M_{i,j}(u)\cup V(T(v))$.

Approximate distances $\tilde d_{i,j}(u,v)$ are maintained by $\mathcal D_{i,j}(u)$ in a way similar to that in Section~\ref{subsec:DetApprox}: $\tilde d_{i,j}(u,v) = (1+\epsilon')^{2i}d_{i,j}$ if the min key value of $Q_{i,j}(u,v)$ is at most $(1+\epsilon')^{2i}d_{i,j}$ and otherwise $\tilde d_{i,j}(u,v) = \infty$.

Data structures $\mathcal D_i(u)$ as well as the overall data structure $\mathcal D$ work exactly as in Section~\ref{subsec:DetApprox}.

\subsection{Correctness}
Consider any point during the sequence of edge deletions. We will show that for suitable choice of $\epsilon'$, we have $d_G(u,v)\leq\tilde d(u,v)\leq (1+\epsilon)d_G(u,v)$.% $Q_i(u,v)\neq\emptyset$ for every vertex pair $(u,v)$.

We will show that during all updates and for all vertex pairs $(u,v)$, if $d_G(u,v)\in (0,D_{i+1}(1+\epsilon')]$ then $\tilde d_i(u,v)\leq (1+\epsilon')^{2i}d_G(u,v)$. By picking $\epsilon' = \ln(1+\epsilon)/(2\lfloor\log_\rho n\rfloor)$, we will then get $d_G(u,v)\leq\tilde d_G(u,v)\leq (1+\epsilon')^{2\lfloor\log_\rho n\rfloor}d_G(u,v)\leq e^{2\lfloor\log_\rho n\rfloor\epsilon'}d_G(u,v)\leq (1+\epsilon)d_G(u,v)$, as desired.

The proof is by induction on $i$. The claim is clear for $i$ with $D_i<33\lg n$ since then we use the data structure of Even and Shiloach. Now, consider an $i$ with $D_i\geq 33\lg n$ and assume that the claim holds for values less than $i$. By the induction hypothesis, we only need to consider pairs $(u,v)$ with $d_G(u,v)\in (D_i(1+\epsilon'),D_{i+1}(1+\epsilon')]$, i.e., $d_G(u,v)\in (d_{i,j},d_{i,j+1}]$ with $j > 0$.

We first show the following invariant for marked vertices that holds prior to each update over the entire sequence of updates:
\begin{Inv}\label{Inv:DynSep}
%For every $w\in M$ with $d_G(u,w)\in (d_{i,j},d_{i,j+1}]$, $j > 0$, if $d_G(u,w)\leq d_{i,j+1}$ then every shortest $u$-to-$w$ path in $G$ intersects a vertex $s\in Q_{i,j}(u,w)$ such that $d_G(u,s)\leq D_i$ and $d_G(s,w)\leq D_i$.
At the end of each update, for every $w\in M_{i,j}(u)$ with $d_G(u,w)\in (d_{i,j},d_{i,j+1}]$, each shortest $u$-to-$w$ path in $G$ intersects a vertex $s\in Q_{i,j}(u,w)$ such that $d_G(u,s)\leq D_i$ and $d_G(s,w)\leq D_i$.
\end{Inv}
\begin{proof}
The invariant is shown by induction on the rank of $w$ in the order in which vertices are marked. Note that this is a proof by induction inside a step of the main proof by induction on $i$; in addition to the induction hypothesis stated above, we may thus assume that the invariant holds for values less than $i$. Additionally, for the current value of $i$, we may assume by induction that the invariant holds for vertices of lower rank than $w$.

Let $t_1$ be a time step with $w\in M_{i,j}(u)^{(t_1)}$ and $d_G^{(t_1)}(u,w)\in (d_{i,j},d_{i,j+1}]$, let $t_0\leq t_1$ be the time step in which $w$ was marked, and let $r$ be the vertex from which an in-tree $T(r)\ni w$ was grown in time step $t_0$. Let $P$ be a shortest $u$-to-$w$ path in $G^{(t_1)}$.

We must have $\tilde d_{i-1}^{(t_0)}(u,r) > D_i(1+\epsilon')^{2i}$ since otherwise, no processing would be done for $r$ in time step $t_0$, contradicting that $T(r)$ is grown in that time step. We also have $d_G^{(t_0)}(u,r) > D_i(1+\epsilon')$ since otherwise the induction hypothesis would give the contradiction $D_i(1+\epsilon')\geq d_G^{(t_0)}(u,r)\geq \tilde d_{i-1}^{(t_0)}(u,r)/(1+\epsilon')^{2(i-1)} > D_i(1+\epsilon')^{2i - 2(i-1)} = D_i(1+\epsilon')^2$.

By the triangle inequality and the fact that $w\in T(r)$ and $T(r)$ has radius at most $\epsilon'D_i$, we get $d_G^{(t_0)}(u,w)\geq d_G^{(t_0)}(u,r) - d_G^{(t_0)}(w,r) > D_i(1+\epsilon') - \epsilon'D_i = D_i$. Hence, $D_i < d_G^{(t_0)}(u,w)\leq d_G^{(t_1)}(u,w)\leq d_{i,j+1}\leq\frac{34}{33}D_i$ so by Lemma~\ref{Lem:DynSep}, $P$ intersects $S_i^{(t_0)}(u)$ and for the first such intersection vertex $s$ along $P$, $d_G^{(t_0)}(u,s)\leq d_G^{(t_1)}(u,s)\leq D_i$ and $d_G^{(t_0)}(s,w)\leq d_G^{(t_1)}(s,w)\leq D_i$. We consider the two cases in the description of $\mathcal D_{i,j}(u)$:
\paragraph{Case $1$:}
It suffices to show that $s\in Q_{i,j}^{(t_1)}(u,w)$.

We have $d_G^{(t_0)}(s,r)\leq d_G^{(t_0)}(s,w) + d_G^{(t_0)}(w,r)\leq (1+\epsilon')D_i$. By the induction hypothesis,
\begin{align*}
\tilde d_{i-1}^{(t_0)}(u,s) + \tilde d_{i-1}^{(t_0)}(s,r) & \leq (1+\epsilon')^{2(i-1)}d_G^{(t_0)}(u,r)\\
  & \leq (1+\epsilon')^{2i-2}(d_G^{(t_0)}(u,w) + \epsilon'D_i)\\
  & \leq (1+\epsilon')^{2i-2}(d_G^{(t_1)}(u,w) + \epsilon'D_i)\\
  & \leq (1+\epsilon')^{2i-2}(d_{i,j+1} + \epsilon'd_{i,j+1})\\
  & = (1+\epsilon')^{2i}d_{i,j},
\end{align*}
so $s\in Q_{i,j}^{(t_0)}(u,w) = Q_{i,j}^{(t_1)}(u,w)$, showing maintenance of the invariant.

\paragraph{Case $2$:}
We first show that $P$ must intersect the set $L$ formed when growing $T(r)$ in time step $t_0$. Since we are in Case $2$, every leaf of $T(r)$ either belongs to $L$ or has no ingoing edges from vertices not in $T(r)$; otherwise, $T(r)$ would contain more than $\epsilon'D_i$ vertices since it is grown up to radius $\epsilon'D_i$. Hence, the only way that $P$ could not intersect $L$ would be if $P$ were fully contained in $T(r)$. But this is not possible since then $T(r)$ would contain at least $|P| + 1\geq d_{i,j} + 1 > D_i\geq \epsilon'D_i$ unmarked vertices at the beginning of time step $t_0$, contradicting that we are in Case $2$.

Thus, $P$ intersects $L$ and we have $w\notin L$ since $w$ was an unmarked vertex of $T(r)$ when growing this tree. Let $x$ be the last vertex of $P$ belonging to $L$. Since $x$ was marked earlier than $w$, the induction hypothesis implies that the subpath of $P$ from $u$ to $x$ intersects $Q_{i,j}^{(t_1)}(u,x) = Q_{i,j}^{(t_0)}(u,x)$ in a vertex $s_x$ such that $d_G^{(t_0)}(u,s_x)\leq d_G^{(t_1)}(u,s_x)\leq D_i$ and $d_G^{(t_0)}(s_x,x)\leq d_G^{(t_1)}(s_x,x)\leq D_i$. The latter implies $d_G^{(t_0)}(s_x,r)\leq (1+\epsilon')D_i$. By the induction hypothesis, $\tilde d_{i-1}^{(t_0)}(u,s_x) + \tilde d_{i-1}^{(t_0)}(s_x,r) \leq (1+\epsilon')^{2(i-1)}(d_G^{(t_0)}(u,s_x) + d_G^{(t_0)}(s_x,r)) = (1+\epsilon')^{2(i-1)}d_G^{(t_0)}(u,r)$ which by the same calculations as in Case $1$ is at most $(1+\epsilon')^{2i}d_{i,j}$. Inspecting the execution of $\mathcal D_{i,j}(u)$ in Case $2$, it follows that $s_x\in Q_{i,j}^{(t_0)}(u,w) = Q_{i,j}^{(t_1)}(u,w)$.
%\todo[inline]{CW: Here, a figure would be useful to illustrate $u$, $s$, $s_x$, $x$, $w$, $T(r)$, etc. It is a lot to keep track of for the reader.}
We have $s_x\in Q_{i,j}^{(t_0)}(u,x)\subseteq S_i^{(t_0)}(u)$. Since $s$ is the first vertex of $S_i^{(t_0)}(u)$ along $P$, $P$ can thus be decomposed into $u\leadsto s\leadsto s_x\leadsto x\leadsto w$ and we get $d_G^{(t_1)}(u,s_x)\leq D_i$ (as shown above) and $d_G^{(t_1)}(s_x,w)\leq d_G^{(t_1)}(s,w)\leq D_i$. This shows maintenance of the invariant with $s_x$ in place of $s$.
\end{proof}
Now, we continue with our proof by induction on $i$. Consider any vertex pair $(u,v)$ at the end of an update with $d_G(u,v)\in (d_{i,j},d_{i,j+1}]$ and $j > 0$.

If $v\notin M_{i,j}(u)$ and $\tilde d_{i-1}(u,v)\leq (1+\epsilon')^{2i}D_i$ then $d_G(u,v)\leq\tilde d_{i,j}(u,v)\leq (1+\epsilon')^{2i}D_i < (1+\epsilon')^{2i}d_G(u,v)$, as desired.

Now assume that $v\notin M_{i,j}(u)$ and $\tilde d_{i-1}(u,v) > (1+\epsilon')^{2i}D_i$. Since $v$ was not marked in the current update, the min key value of $Q_{i,j}(u,v)$ at the end of the update is at most $d_{i,j}(1+\epsilon')^{2i}$ so $d_G(u,v)\leq\tilde d_{i,j}(u,v)\leq (1+\epsilon')^{2i}d_{i,j} < (1+\epsilon')^{2i}d_G(u,v)$, as desired.

Finally assume that $v\in M_{i,j}(u)$. By Invariant~\ref{Inv:DynSep}, there is an $s\in Q_{i,j}(u,v)$ such that $d_G(u,v) = d_G(u,s) + d_G(s,v)$, $d_G(u,s)\leq D_i$, and $d_G(s,v)\leq D_i$. By the induction hypothesis, $d_G(u,v)\leq \tilde d_i(u,v)\leq\tilde d_{i-1}(u,s) + \tilde d_{i-1}(s,v)\leq (1+\epsilon')^{2(i-1)}d_G(u,v)$, as desired. This completes the inductive proof and correctness follows.

%Next, consider the case where $|V(T(v'))|\leq \epsilon'd_i$ and consider any $s\in S_i(u,v)\setminus Q_{i,j}^{(t_0)}(u,v)$. Consider a path $P$ in $G^{(t_0)}$ consisting of a shortest path $P_1$ from $u$ to $s$ followed by a shortest path $P_2$ from $s$ to $v$. There is a vertex $w\in L\cap P_2$ and since $s\notin Q_{i,j}^{(t_0)}(u,v)$ then in particular, $s\notin Q_{i,j}^{(t_0)}(u,w)$. But since $w$ is marked, $\tilde d_i^{(t_1)}(u,s) + \tilde d_i^{(t_1)}(s,w)\geq \tilde d_i^{(t_0)}(u,s) + \tilde d_i^{(t_0)}(s,w)\geq d_{i,j+1}$.

%It follows that there is an $s\in Q_{i,j}^{(t_1)}(u,v)$ such that $d_{i,j} < \tilde d_i^{(t_1)}(u,s) + \tilde d_i^{(t_1)}(s,v)\leq (1+\epsilon')^{i+1}d_{i,j}$ so $\tilde d_{i,j}^{(t_1)}(u,v) = (1+\epsilon')^{i+1}d_{i,j}$, as desired.

\subsection{Running time}
Maintaining separators $S_i(u)$ over all $u$ and $i$ takes $O(mn\log_{\rho} n) = O(mn\log n)$ time by Lemma~\ref{Lem:DynSep}. For the remaining time analysis, we focus on a single data structure $\mathcal D_{i,j}(u)$. It is useful in the following to regard this structure as handling an adversarial sequence of updates consisting of changes to approximate distances maintained by structures $\mathcal D_{i'}(v)$ for $i' < i$ and $v\in V$. We will give an expected time bound for $\mathcal D_{i,j}(u)$ and we shall rely on the following key lemma.
\begin{Lem}\label{Lem:SizeZeroSamples}
Let $r\in V$. If at some point in the sequence of updates, $\mathcal D_{i,j}(u)$ grows an in-tree from $r$ then at the end of that update, the expected number of vertices $s\in S_i(u)$ satisfying $\tilde d_{i-1}(u,s) + \tilde d_{i-1}(s,r)\leq D_i(1+\epsilon')^{2i}$ is $O(\ln n/p)$. This bound holds against an adaptive adversary.
\end{Lem}
\begin{proof}
Assume that an in-tree is grown from $r$ at some time step $t_1$. Then we must have $\tilde d_{i-1}^{(t_1)}(u,r) > D_i(1+\epsilon')^{2i}$. Let $t_0\leq t_1$ be the earliest time step where $\tilde d_{i-1}^{(t_0)}(u,r) > D_i(1+\epsilon')^{2i}$. For the analysis, consider a modification $\mathcal D_{i,j}'(u)$ of $\mathcal D_{i,j}(u)$ which when processing $r$ in each update $t\in\{t_0,\ldots,t_2\}$ applies a deterministic algorithm that maintains $Q_{i,j}^{(t)}(u,r)$ as the set of all vertices $s\in S_i^{(t)}(u)$ with corresponding key values $\tilde d_{i-1}^{(t)}(u,s) + \tilde d_{i-1}^{(t)}(s,r)$; here $t_2$ is the time step of $\mathcal D_{i,j}'(u)$ in which $r$ is marked. Note that from time step $t_0$ to $t_2-1$, the min key value of $Q_{i,j}(u,r)$ is at most $d_{i,j}(1+\epsilon')^{2i}$. Thus, $t_1\leq t_2$ and $\mathcal D_{i,j}'(u)$ and $\mathcal D_{i,j}(u)$ maintain exactly the same approximate distances for each time step $t\in\{t_0,\ldots,t_1-1\}$, namely $\tilde d_i^{(t)}(u,r) = d_{i,j}(1+\epsilon')^{2i}$. Hence, the output revealed to the adversary during these updates is the same when using $\mathcal D_{i,j}'(u)$ as when using $\mathcal D_{i,j}(u)$. From this and from the observation that the update done by the adversary in time step $t_1$ only depends on outputs from earlier updates, the sequence of updates from time step $t_0$ to $t_1$ is independent of which vertices of $S_i(u)$ are sampled.

The above relates to the following experiment. We have a dynamic set $X$ undergoing a sequence of $k$ updates consisting of insertions and deletions of single elements. For $t = 0,\ldots,k$, let $n_t$ denote the number of elements in $X$ after the $t$'th update. Associated with $X$ we have a dynamic subset $Y$. The initial set $Y$ is obtained by sampling each element of the initial set $X$ independently with some probability $q$. Whenever a new element is inserted into $X$, it is added to $Y$ with probability $q$. We assume that the sequence of updates to $X$ is independent of which elements are sampled. Then for $t = 0,\ldots,k$, $\Pr[Y^{(t)} = \emptyset] = (1-q)^{n_t}$.

Returning to the analysis of our algorithm, let $q = p$ and let $X$ denote the dynamic subset of $S_i(u)$ consisting of elements $s$ with $\tilde d_{i-1}(u,s) + \tilde d_{i-1}(s,r) \leq D_i(1+\epsilon')^{2i}$ during time steps $t_0,\ldots,t_2$. Note that $Y^{(t_1)} = \emptyset$. Let $T = \{t_0,\ldots,t_2\}$. For each $t\in T$,
\[
 \Pr[t_1 = t]\leq \Pr[Y^{(t)} = \emptyset] = (1-p)^{|X^{(t)}|}]\leq e^{-p|X^{(t)}|}.
\]
Let $T_1$ be the set of elements $t\in T$ where $|X^{(t)}|\leq \ln((m+1)n)/p$ and let $T_2 = T\setminus T_1$. Since $|T_2|\leq |T|\leq m+1$ and since $|X^{(t)}|\leq n$ for each $t\in T$,
\begin{align*}
  E[|X^{(t_1)}|] & =    \sum_{t\in T_1}|X^{(t)}|\Pr[t_1 = t] + \sum_{t\in T_2}|X^{(t)}|\Pr[t_1 = t]\\
               & \leq \ln((m+1)n)/p\sum_{t\in T_1}\Pr[t_1 = t] + \sum_{t\in T_2}|X^{t}|e^{-p|X^{(t)}|}\\
               & \leq 3\ln n/p + \sum_{t\in T_2}|X^{(t)}|/((m+1)n)\\
               & \leq 3\ln n/p + 1\\
               & = O(\ln n/p),
\end{align*}
as desired.
\end{proof}

\begin{Cor}\label{Cor:SizeZeroSamples}
When a vertex $v$ is marked, $E[|Q_{i,j}(u,v)|] = O(\ln n/p)$ and this bound holds against an adaptive adversary.
\end{Cor}
\begin{proof}
Consider the update in which $v$ is marked and let $r$ be the root of the in-tree $T(r)$ containing $v$. If $|T(r)|\geq \epsilon'D_i$ then $Q_{i,j}(u,v) = Q_{i,j}(u,r)\subseteq S_i(u)$ and all $s\in Q_{i,j}(u,r)$ satisfy the inequality of Lemma~\ref{Lem:SizeZeroSamples}. In the case where $|T(r)| < \epsilon'D_i$, let $Q$ be as defined in the description of the data structure. Then vertices $s\in Q\subseteq S_i(u)$ are only added to $Q_{i,j}(u,v)$ if they satisfy the inequality of Lemma~\ref{Lem:SizeZeroSamples}. The corollary now follows.
\end{proof}

Now, we can bound the time spent by $\mathcal D_{i,j}(u)$. The total time spent on growing in-trees is $O(m)$ since every edge $(w_1,w_2)$ visited must have $w_1\notin M_{i,j}(u)$ at the beginning of the BFS search and $w_1\in M_{i,j}(u)$ immediately afterwards and a vertex can never be unmarked. This also bounds the time spent on marking vertices.

The total expected number of sampled vertices added to $Q_{i,j}(u,v)$ prior to $v$ being marked is at most $xp$ where $x$ is the size of the set $S_i(u)$ after the final update. By Lemma~\ref{Lem:DynSep}, $x = O(n\log n/D_i)$. By Corollary~\ref{Cor:SizeZeroSamples}, the expected size of $Q_{i,j}(u,v)$ after $v$ is marked is $O(\ln n/p)$. Using the same argument as in the running time analysis of Section~\ref{subsec:DetApprox}, the number of increase-key operations applied to a single element of $Q_{i,j}(u,v)$ is $O(\log n/\epsilon')$. Hence, the total expected time spent on operations on $Q_{i,j}(u,v)$ is $O((n\log n\cdot p/D_i + \log n/p)\log^3n/\epsilon)$.

Whenever $\mathcal D_{i,j}(u)$ grows an in-tree $T(r)$ with $|V(T(r))\setminus M_{i,j}(u,v)| > \epsilon'D_i$, scanning $S_i(u)$ takes $O(n\log n/D_i)$ time by Lemma~\ref{Lem:DynSep}. Since all vertices of $V(T(r))\setminus M_{i,j}(u,v)$ are marked just after $T(r)$ is grown and since vertices are never unmarked, the number of such trees over the course of the updates is at most $n/(\epsilon'D_i)$ so the total time for all these scans is $O(n^2\log^2n/(\epsilon D_i^2))$.

Whenever $\mathcal D_{i,j}(u)$ grows an in-tree $T(r)$ with $|V(T(r))\setminus M_{i,j}(u,v)| \leq \epsilon'D_i$, the set $\cup_{x\in L}Q_{i,j}(u,x)$ needs to be computed. Note that for each $x\in L$, $E[|Q_{i,j}(u,x)|] = O(\log n/p)$ by Corollary~\ref{Cor:SizeZeroSamples}. At least one edge $(y,x)$ ingoing to $x$ belongs to $T(r)$ and this edge is not part of any later grown in-tree since $x$ is marked immediately after $T(r)$ is grown. We charge a cost of $O(\log n/p)$ to $(y,x)$ for computing $Q_{i,j}(u,x)$. Over all $x\in L$, this pays for computing $\cup_{x\in L}Q_{i,j}(u,x)$ and we get a total expected time bound for this part of $O(m\log n/p)$.

Summing the above over all $u$, $v$, $i$, and $j$, we get a total expected time bound for our data structure of
\[
  \tilde O(mn/\epsilon + \sum_i\sum_j (n^3\cdot p/(D_i\epsilon) + n^2/(p\epsilon) + n^3/(\epsilon D_i^2) + mn/p).
\]

Since this bound is only fast for sufficiently large $i$, we pick a distance threshold $d$ and apply the algorithm of Even and Shiloach for distances of at most $d$ and our data structure for distances above $d$. By a geometric sums argument, our hybrid algorithm has a expected total time bound of
\begin{align*}
  & \tilde O(mnd + mn/\epsilon + n^3\cdot p/(d\epsilon^2) + n^2/(p\epsilon^2) + n^3/(\epsilon^2d^2) + mn/(p\epsilon))\\
  & = \tilde O(mnd + n^3\cdot p/(d\epsilon^2) + n^2/(p\epsilon^2) + n^3/(d^2\epsilon^2) + mn/(p\epsilon))
\end{align*}
Setting the second and fifth terms equal to each other, we get $p = \tilde\Theta(\sqrt{m\epsilon d}/n)$ and the time bound simplifies to
\[
  \tilde O(mnd + \sqrt{m}n^2/(\sqrt d\epsilon^{3/2}) + n^3/(\sqrt{md}\epsilon^{5/2}) + n^3/(d^2\epsilon^2)).
\]
We balance the first two terms by setting $d = \tilde\Theta(n^{2/3}/(m^{1/3}\epsilon))$ and we get a time bound of
\[
  \tilde O(m^{2/3}n^{5/3}n/\epsilon + n^{8/3}/(m^{1/3}\epsilon^2)),
\]
which shows the time bound of Theorem~\ref{Thm:Rand}.

%%
%% Bibliography
%%
\printbibliography[heading=bibintoc] % Make bibliography show up in table of contents

\appendix

\section{Related Work} \label{sec:relatedWork}

\paragraph{Undirected APSP.} In undirected graphs, maintaing all-pairs shortest paths with stetch $(1+\eps)$ was first studied by Roditty and Zwick  \cite{Roditty2004} who achieved total running time $\tilde{O}(mn)$ for unweighted graphs. Their data structure was then derandomized by Henzinger et al. \cite{henzinger2016dynamic} which was in turn simplified in \cite{gutenberg2020deterministic}. In \cite{henzinger2016dynamic}, the authors also give a $(2+\eps)$-approximate all-pairs shortest path for undirected, unweighted graphs with total update time $\tilde{O}(n^{2.5}/\eps)$. For high stretch, Henzinger et al. \cite{Henzinger2014} gave an oblivious data structure with total update time $O(mn^{1/k + o(1)} \text{polylog } W)$ with stretch $O(k^k)$ for any integer $k \geq 1$. Chechik \cite{Chechik2018} obtained the same running time but obtained near-optimal stretch $2(1+\eps)k - 1$. Very recently, Chuzhoy and Saranurak \cite{chuzhoy2020deterministic} gave the first \emph{deterministic} decremental all pair shortest paths algorithm that achieves subcubic running time for any graph density. However, their algorithm only achieves a large constant approximation factor.

Finally, we also point out that there is a $(2+\eps)$-approximate all-pairs shortest paths by Bernstein \cite{bernstein2009fully} with $\tilde{O}(m)$ \emph{amortized} update time for \emph{fully-dynamic} graphs and that Abraham et al. \cite{abraham2013dynamic} showed that sublinear \emph{amortized} update time is possible for constant stretch.

\paragraph{APSP with Worst-case Update Time.} For the fully-dynamic setting, Thorup also introduced the problem of maintaining APSP with \emph{worst-case} update time. In \cite{thorup2005worst}, he presents a deterministic data structure that achieves $\tilde{O}(n^{2.75})$ worst-case time. This bound was recently improved to $\tilde{O}(n^{2.66})$ by Abraham et al. \cite{AbrahamCK17} who presented an adaptive randomized algorithm. Recently, Probst Gutenberg and Wulff-Nilsen \cite{gutenberg2020fully} gave a deterministic algorithm that breaks the $\tilde{O}(n
^{2.75})$ update time bound by Thorup.

\paragraph{Single-Source Shortest Path.} For the single-source shortest-path problem, a recent line of research \cite{bernstein2011improved, henzinger2014subquadratic} has culminated in a $(1+\eps)$-approximate data structure by Henzinger et al. \cite{Henzinger2014} for partially-dynamic \emph{undirected graphs}. They achieve total update time $m^{1+o(1)}$, however, they need to asssume an oblivious adversary. To overcome this restriction, Bernstein and Chechik, recently introduced a framework \cite{bernstein2016deterministic, bernstein2017deterministic, bernstein2017deterministicweighted} to maintain $(1+\eps)$-approximate shortest paths against an adaptive adversary in time $\tilde{O}(\min\{n^2, mn^{3/4}\})$. Even more recently, Chuzoy and Khanna \cite{Chuzhoy:2019:NAD:3313276.3316320} extended their framework and showed that it can be used to improve the static problems of vertex-capacitated max-flow and sparsest vertex cut. Probst Gutenberg and Wulff-Nilsen \cite{gutenberg2020deterministic} recently presented a deterministic algorithm that improves on the former bounds for sparse graphs with total update time $mn
^{0.5+o(1)}$. The existing data structures where futher extended in \cite{bernstein2020fully} to be path-reporting.

In the directed, weighted setting, Henzinger et al. \cite{henzinger2014sublinear} presented a $(1+\eps)$-approximate data structure with total update time $\tilde{O}(mn^{0.9+o(1)})$ for the decremental setting. A new approach by Bernstein et al. \cite{GutenbergW20a, nearOptDenseSSSP} has recently obtained running time $\tilde{O}(n^2, mn^{2/3})$ for decremental weighted digraphs, which is near-optimal when the graph is dense e.g. $m = \Theta(n^2)$. The simpler problem of maintain Single-Source Reachability in a decremental digraph was further solved to near-otimality \cite{bernstein2019decremental}. For Decremental Single-Source Reachability and SSSP, deterministic algorithms that improve over the classic ES-tree were given by Bernstein et al. \cite{detDiSSSP}. For the incremental setting, Probst Gutenberg et al. \cite{GutenbergWW20} recently obtained a deterministic $(1+\eps)$-approximate algorithm with total update time $\tilde{O}(n
^2)$. 

In both settings, the exact partially-dynamic SSSP problem was proven to require $\Omega(mn^{1-o(1)})$ total update time \cite{Roditty2004,abboud2014popular, henzinger2015unifying, GutenbergWW20}, assuming in various popular conjectures.

\section{Sampling technique} \label{app:sampling}
Here we describe how we use the sampling technique used by Wulff-Nilsen \cite{wulff2017fully} lemma 28 when $\mathcal{D}_{i, j}(u)$ adds vertices to $S_{i, j}$. Remember that when a vertex $s$ is added to $S_i(u)$ it is added to $S_{i, j}(u, v)$ with probability $p$ for every $v \in V$, but since flipping a coin for every vertex is too slow, we do something different to simulate that process:

Enumerate the vertices arbitrarily $v_1 \hdots v_n$ and let $\mathcal{E}_{k,s}$ denote the event that $v_k$ is the first to sample a vertex $s \in S_i(u)$. Clearly $p_{k, s} = \Pr[\mathcal{E}_{k,s}] = (1-p)^{k-1}p$ for every $k \leq n$, and $p_{n+1, s} = (1-p)^{n}$ i.e. the vertex wasn't added to any $S_{i, j}(u, v)$. The algorithm can then sample from this distrubtion to pick $k$ and then recurse on the remaining $n-k$ vertices.

It remains to show how to sample from this distribution. let $p_{i_1, i_2, s} = \sum_{i = i_1}^{i_2}p_{i, s}$ i.e. the probability that the we want $k$ is between $i_1$ and $i_2$. We precompute these values for $1 \leq i_1 \leq i_2 \leq n+1$ in time $O(n^2)$ using dynamic programming, note that these numbers are the same for every $\mathcal{D}_{i, j}$ so we only need to compute them once. 

Using these numbers we sample $k$ as follows: Start with $i_1 = 1$, $i_2 = n+1$ and $j = \lceil i_1 + i_2 / 2 \rceil$ With probabilty $\frac{p_{i_1, j, s}}{p{i_1, i_2, s}}$ set $i_2 = j$ and $i_1 = j+1$ otherwise, and recurse like binary search would until $i_1 = i_2 = k$. Clearly we can find $k$ in logarithmic time using the precomputed probabilities, so for every sampled vertex we can charge $O(\log n)$ running time to every $\mathcal{D}_{i, j}(u)$ for which the vertex was added to $S_{i, j}(u, v)$, in total $O(|S_i(u)|p \log n)$ for every  $\mathcal{D}_{i, j}(u)$.

\end{document}